\documentclass[draftcls, 12pt, onecolumn, twoside]{IEEEtran}

\usepackage{amsfonts}
\usepackage{amsmath}
\usepackage{amssymb}
\usepackage{lscape}
\usepackage{epsf}
\usepackage{graphicx}

\newtheorem{theorem}{\indent Theorem}[section]
\newtheorem{lemma}[theorem]{\indent Lemma}
\newtheorem{corollary}[theorem]{\indent Corollary}
\newtheorem{proposition}[theorem]{\indent Proposition}

\newtheorem{EXAMPLE}{\indent Example}[section]
\newtheorem{definition}{\indent Definition}[section]

\newenvironment{example}{\begin{EXAMPLE}\rm}{\rm\end{EXAMPLE}}

\newcommand{\code}{{\mathcal{C}}}
\newcommand{\graph}{{\mathcal{G}}}
\newcommand{\cL}{{\mathcal{L}}}
\newcommand{\cV}{{\mathcal{V}}}
\newcommand{\cH}{{\mathcal{H}}}
\newcommand{\cF}{{\mathcal{F}}}
\newcommand{\cG}{{\mathcal{G}}}
\newcommand{\ngbr}{{\mathcal{N}}} 
\newcommand{\cE}{{\mathcal{E}}}

\newcommand{\cI}{{\mathcal{I}}}
\newcommand{\cJ}{{\mathcal{J}}}
\newcommand{\cM}{{\mathcal{M}}}
\newcommand{\cN}{{\mathcal{N}}}
\newcommand{\cK}{{\mathcal{K}}}
\newcommand{\cP}{{\mathcal{P}}}
\newcommand{\cQ}{{\mathcal{Q}}}
\newcommand{\cS}{{\mathcal{S}}}
\newcommand{\cT}{{\mathcal{T}}}
\newcommand{\cU}{{\mathcal{U}}}
\newcommand{\cX}{{\mathcal{X}}}

\newcommand{\sG}{{\mathsf{G}}} 
\newcommand{\sV}{{\mathsf{V}}} 
\newcommand{\sU}{{\mathsf{U}}} 
\newcommand{\sE}{{\mathsf{E}}} 
\newcommand{\sT}{{\mathsf{T}}} 
\newcommand{\sX}{{\mathsf{X}}} 
\newcommand{\sY}{{\mathsf{Y}}} 
\newcommand{\sZ}{{\mathsf{Z}}} 
\newcommand{\scp}{{\mathsf{c}}} 
\newcommand{\sfn}{{\mathsf{f}}}

\newcommand\rr{{\mathbb R}}
\newcommand\qq{{\mathbb Q}}

\newcommand\nn{{\mathbb N}}

\newcommand{\GF}{{\mathrm{GF}}}

\newcommand{\blda}{{\mbox{\boldmath $a$}}}
\newcommand{\bldaa}{{\mbox{\scriptsize \boldmath $a$}}}
\newcommand{\bldb}{{\mbox{\boldmath $b$}}}
\newcommand{\bldbb}{{\mbox{\scriptsize \boldmath $b$}}}
\newcommand{\bldc}{{\mbox{\boldmath $c$}}}
\newcommand{\bldcc}{{\mbox{\scriptsize \boldmath $c$}}}
\newcommand{\bldchi}{{\mbox{\boldmath $\chi$}}}

\newcommand{\blddelta}{{\mbox{\boldmath $\delta$}}}
\newcommand{\bldf}{{\mbox{\boldmath $f$}}}
\newcommand{\bldff}{{\mbox{\scriptsize \boldmath $f$}}}
\newcommand{\bldg}{{\mbox{\boldmath $g$}}}
\newcommand{\bldgg}{{\mbox{\scriptsize \boldmath $g$}}}

\newcommand{\bldmu}{{\mbox{\boldmath $\mu$}}}
\newcommand{\bldh}{{\mbox{\boldmath $h$}}}

\newcommand{\bldk}{{\mbox{\boldmath $k$}}}
\newcommand{\bldkk}{{\mbox{\scriptsize \boldmath $k$}}}
\newcommand{\bldkappa}{{\mbox{\boldmath $\kappa$}}}
\newcommand{\bldsubkappa}{{\mbox{\scriptsize \boldmath $\kappa$}}}
\newcommand{\bldm}{{\mbox{\boldmath $m$}}}
\newcommand{\bldp}{{\mbox{\boldmath $p$}}}
\newcommand{\bldr}{{\mbox{\boldmath $r$}}}
\newcommand{\bldrr}{{\mbox{\scriptsize \boldmath $r$}}}

\newcommand{\bldt}{{\mbox{\boldmath $t$}}}

\newcommand{\bldppp}{{\mbox{\scriptsize \boldmath $p$}}}
\newcommand{\bldtt}{{\mbox{\scriptsize \boldmath $t$}}}
\newcommand{\bldtau}{{\mbox{\boldmath $\tau$}}}

\newcommand{\bldw}{{\mbox{\boldmath $w$}}}

\newcommand{\bldx}{{\mbox{\boldmath $x$}}}
\newcommand{\bldxx}{{\mbox{\scriptsize \boldmath $x$}}}

\newcommand{\bldy}{{\mbox{\boldmath $y$}}}
\newcommand{\bldyy}{{\mbox{\scriptsize \boldmath $y$}}}
\newcommand{\bldz}{{\mbox{\boldmath $z$}}}
\newcommand{\bldxi}{{\mbox{\boldmath $\xi$}}}
\newcommand{\bldXi}{{\mbox{\boldmath $\Xi$}}}
\newcommand{\bldXiXi}{{\mbox{\scriptsize \boldmath $\Xi$}}}
\newcommand{\bldG}{{\mbox{\boldmath $G$}}}
\newcommand{\bldL}{{\mbox{\boldmath $L$}}}
\newcommand{\bldP}{{\mbox{\boldmath $P$}}}
\newcommand{\bldlambda}{{\mbox{\boldmath $\lambda$}}}
\newcommand{\bldsigma}{{\mbox{\boldmath $\sigma$}}}
\newcommand{\ones}{{\mbox{\boldmath $1$}}}
\newcommand{\zeros}{{\mbox{\boldmath $0$}}}
\newcommand{\rrr}{\mathfrak{R}}%
\newcommand{\ppp}{\mathfrak{P}}%
\newcommand{\zeross}{{\mbox{\scriptsize \boldmath $0$}}}
\newcommand{\rrrm}{\rrr^{-}}

\newcommand{\bldzero}{{\mbox{\boldmath $0$}}}

\newcommand{\Prob}{{p}}

\newcommand{\half}{{\textstyle\frac{1}{2}}}

    \def\squarebox#1{\hbox to #1{\hfill\vbox to #1{\vfill}}}

\newlength{\Algwidth}

\title{Linear-Programming Decoding of Nonbinary Linear Codes    
\thanks{%
    This work was supported in part by the Claude Shannon Institute for 
    Discrete Mathematics, Coding and Cryptography (Science Foundation Ireland Grant 06/MI/006). The material in this paper was presented in part at the 7-th International ITG Conference on Source and Channel Coding (SCC), Ulm, Germany, January 2008, and in part at the IEEE International Symposium on Information Theory (ISIT), Toronto, Canada, July 2008.  
    \newline 
    M. F. Flanagan is with the School of Electrical, Electronic and Mechanical Engineering, University College Dublin, Belfield, Dublin 4, Ireland (e-mail:mark.flanagan@ieee.org).
    \newline
    V. Skachek, E. Byrne and M. Greferath are with the Claude Shannon Institute
    and the School of Mathematical Sciences, University College Dublin, Belfield, Dublin 4, Ireland 
    (e-mail:\{vitaly.skachek, ebyrne, marcus.greferath\}@ucd.ie). 
    }        
}
\author{Mark F. Flanagan,~\IEEEmembership{Member,~IEEE,} Vitaly Skachek,~\IEEEmembership{Member,~IEEE,} 
Eimear Byrne, and Marcus Greferath} 
\begin{document}
\maketitle

\begin{abstract}
A framework for linear-programming (LP) decoding of nonbinary linear codes over rings is developed. This framework facilitates linear-programming based reception for coded modulation systems which use direct modulation mapping of coded symbols. It is proved that the resulting LP decoder has the `maximum-likelihood certificate' property. It is also shown that the decoder output is the lowest cost pseudocodeword. Equivalence between pseudocodewords of the linear program and pseudocodewords of graph covers is proved. It is also proved that if the modulator-channel combination satisfies a particular symmetry condition, the codeword error rate performance is independent of the transmitted codeword. Two alternative polytopes for use with linear-programming decoding are studied, and it is shown that for many classes of codes these polytopes yield a complexity advantage for decoding. These polytope representations lead to polynomial-time decoders for a wide variety of classical nonbinary linear codes. LP decoding performance is illustrated for the $[11,6]$ ternary Golay code with ternary PSK modulation over AWGN, and in this case it is shown that the performance of the LP decoder is comparable to codeword-error-rate-optimum hard-decision based decoding. LP decoding is also simulated for medium-length ternary and quaternary LDPC codes 
with corresponding PSK modulations over AWGN. 

\textbf{Keywords:}
Linear-programming decoding,
LDPC codes,
pseudocodewords,
coded modulation. 
\end{abstract}

\section{Introduction}

\emph{Low-density parity-check} (LDPC) codes~\cite{Gallager} have become very popular in recent years due to their excellent performance under \emph{sum-product} (SP) decoding (or \emph{message-passing} decoding). The primary research focus in this area to date has been on \emph{binary} LDPC codes. Finite-length analysis of such LDPC codes under SP decoding is a difficult task. An approach to such an analysis was proposed in~\cite{Wiberg} based on the consideration of so-called \emph{pseudocodewords} and their \emph{pseudoweights}, defined with respect to a structure called the \emph{computation tree}. By replacing this set of pseudocodewords with another set defined with respect to cover graphs of the Tanner graph (here called \emph{graph-cover pseudocodewords}), the analysis was found to be significantly more tractable while still yielding accurate experimental results~\cite{FKKR},~\cite{KV-characterization},~\cite{KV-IEEE-IT}. 

In~\cite{Feldman-thesis} and~\cite{Feldman}, the decoding of \emph{binary} LDPC codes using \emph{linear-programming} (LP) decoding was proposed, and many important connections between linear-programming decoding and classical message-passing decoding were established. In particular, it was shown that the LP decoder is inhibited by a set of pseudocodewords corresponding to points in the LP relaxation polytope with rational coordinates (here called \emph{linear-programming pseudocodewords}), and that the set of these pseudocodewords is equivalent to the set of graph-cover pseudocodewords. This represents a major result as it indicates that essentially the same phenomenon determines performance of LDPC codes under both LP and SP decoding.

For high-data-rate communication systems, bandwidth-ef\-fi\-cient signalling schemes are required which necessitate the use of higher-order  (nonbinary) modulation. Of course, within such a framework it is desirable to use state-of-the-art error-correcting codes. Regarding the combination of LDPC coding and higher-order modulation, bit-interleaved coded modulation (BICM)~\cite{Caire} is a high-performance method which cascades the operations of binary coding, interleaving and higher-order constellation mapping. Here however the problem of system analysis is exacerbated by the complication of joint design of binary code, interleaver and constellation mapping; this becomes even more difficult when feedback is included from the decoder to the demodulator~\cite{Ritcey}. 

Alternatively, higher-order modulation may be achieved in conjunction with coding by the use of nonbinary codes whose symbols map directly to modulation signals. A study of such codes over rings, for use with PSK modulation, was performed in~\cite{Deepak}, with particular focus on the ring of integers modulo $8$. Nonbinary LDPC codes over fields have been investigated with direct mapping to binary~\cite{MacKay} and nonbinary~\cite{LSLG},~\cite{Bennatan1},~\cite{Bennatan2},~\cite{Bennatan-thesis} modulation signals; in all of this work, SP decoding (with respect to the nonbinary alphabet) was assumed. Recently, some progress has been made on the topic of analysis of such codes; in particular, pseudocodewords of nonbinary codes were defined and some bounds on the pseudoweights were derived ~\cite{Kelley-Sridhara-ISIT-2006}. 
 
In this work, we extend the approach in~\cite{Feldman} towards coded modulation, in particular to codes over rings mapped to nonbinary modulation signals. As was done in~\cite{Feldman}, we show that the problem of decoding may be formulated as an LP problem for the nonbinary case. We also show that an appropriate relaxation of the LP leads to a solution which has the `maximum-likelihood (ML) certificate' property, i.e. if the LP outputs a codeword, then it must be the ML codeword. Moreover, we show that if the LP output is integral, then it must correspond to the ML codeword. We define the \emph{graph-cover pseudocodewords} of the code, and the \emph{linear-programming pseudocodewords} of the code, and prove the equivalence of these two concepts. This shows that the links between LP decoding on the relaxation polytope and message-passing decoding on the Tanner graph generalize to the nonbinary case. Of course, while we use the term `nonbinary' throughout this paper, our framework includes the binary framework as a special case.

For coded modulation systems using maximum-likelihood (ML) decoding, the concept of \emph{geometric uniformity} \cite{Forney} was introduced as a condition which, if satisfied, guarantees codeword error rate (WER) performance independent of the transmitted codeword (this condition was used for design of the coded modulation systems in \cite{Deepak}). An analogous symmetry condition was defined in \cite{Richardson} for binary codes over $GF(2)$ with SP decoding; this was later extended to nonbinary codes over $GF(q)$ by invoking the concept of \emph{coset} LDPC codes~\cite{Bennatan1},~\cite{Bennatan2}. We show that for the present framework, there exists a symmetry condition under which the codeword error rate performance is independent of the transmitted codeword. This provides a condition somewhat akin to geometric uniformity for the present framework. It is noteworthy that the same symmetry condition has recently been shown to yield codeword-independent decoder performance in the context of SP decoding \cite{Flanagan_ISIT_08} and also in the context of ML decoding \cite{Hof_Sason_Shamai}. In particular, this identifies a `natural' mapping for nonbinary codes mapped to PSK modulation, where LP, SP or ML decoding is used with direct modulation mapping of coded symbols.

For the binary framework, alternative polytope representations were studied which gave a complexity advantage in certain scenarios~\cite{Feldman-thesis},~\cite{Feldman},~\cite{Chertkov},~\cite{Feldman-Yang},~\cite{Feldman-Yang2}. Analogous to these works, we define 
two alternative polytope representations, which offer a smaller number of variables and constraints for many classes of nonbinary codes. 
We compare these representations with the original polytope, and show that both of them have equal error-correcting performance to the original LP relaxation. Both of these representations lead to polynomial-time decoders for a wide variety of classical nonbinary linear codes.

To demonstrate performance, LP decoding is simulated for the ternary Golay code mapped to ternary PSK over AWGN, and the LP decoder is seen to perform approximately as well as codeword-error-rate optimum hard-decision decoding, and approximately $1.5$ dB from the union bound for codeword-error-rate optimum soft-decision decoding. 

The paper is organized as follows. Section~\ref{sec:general} introduces general settings and notation. The nonbinary decoding problem is formulated as a linear-programming problem in Section~\ref{sec:lp}, and basic properties of the decoding polytope are studied in Section~\ref{sec:pol-properties}.  
A sufficient condition for codeword-independence performance of the decoder is presented in Section~\ref{sec:codeword_independence}. 
\emph{Linear-programming pseudocodewords} are defined in Section~\ref{sec:lp-pseudocodewords}, and their 
properties are discussed. Their equivalence to the \emph{graph-cover pseudocodewords} 
is shown in Section~\ref{sec:eq-pseudocodewords}. Two alternative polytope representations are presented in Sections~\ref{sec:alternative_polytope} and~\ref{sec:cascaded_polytope}, both of which have equivalent performance to the original but may provide lower-complexity decoding. Simulation results are presented in Section~\ref{sec:sims} for some example coded modulation systems. Finally, some directions for future research are proposed in Section~\ref{sec:future}. 
 
\section{General Settings}
\label{sec:general} 

We consider codes over finite rings (this includes codes over finite fields, but may be more general). Denote by $\rrr$ a ring with $q$ elements, by $0$ its additive identity, and let $\rrrm = \rrr \backslash \{ 0 \}$. Let $\code$ be a code of length $n$ over the ring $\rrr$, defined by
\begin{equation}
\code = \{ \bldc \in \rrr^n \; : \; \bldc \cH^T = \bldzero \}
\label{eq:code_definition}
\end{equation}
where $\cH$ is an $m \times n$ matrix (with entries from $\rrr$) called the \emph{parity-check matrix} of the code $\code$. Obviously, the code $\code$ may admit more than one parity-check matrix; however we will consider that the parity-check matrix $\cH$ is fixed in this paper. 

Linearity of the code $\code$ follows directly from (\ref{eq:code_definition}). 
Also, the \emph{rate} of the code $\code$ is defined as $R(\code) = \log_q(\left| \code \right|)/n$ and is equal to the number of information symbols per coded symbol. The code $\code$ may then be referred to as an $[n,\log_q(\left| \code \right|)]$ linear code over $\rrr$.

Denote the set of column indices and the set of row indices of $\cH$ by  $\cI = \{1, 2, \cdots, n \}$ 
and $\cJ = \{1, 2, \cdots, m \}$, respectively. 
We use notation $\cH_j$ for the $j$-th row of $\cH$, where $j \in \cJ$. 
Denote by $\mbox{supp}(\bldc)$ the support of a vector $\bldc$. For each $j\in\cJ$, let $\cI_j = \mbox{supp}(\cH_j)$ and $d_j=|\cI_j|$, and let $d = \max_{j \in \cJ} \{ d_j \}$. 

Given any $\bldc \in \rrr^n$, we say that parity-check $j\in\cJ$ is \emph{satisfied} by $\bldc$ if and only if
\begin{equation}
\bldc \cH_j^T = \sum_{i\in\cI_j} c_i \cdot \cH_{j,i} = 0 \; .
\label{eq:parity_check_satisfied}
\end{equation}
For $j\in\cJ$, define the single parity-check code $\code_j$ over $\rrr$ by 
\[
\code_j = \{ ( b_i )_{i\in\cI_j} \; : \; \sum_{i\in\cI_j} b_i \cdot \cH_{j,i} = 0 \}
\]
Note that while the symbols of the codewords in $\code$ are indexed by $\cI$, the symbols of the codewords in $\code_j$ are indexed by $\cI_j$. 
We define the projection mapping for parity-check $j\in\cJ$ by
\[
\bldx_j(\bldc) = ( c_i )_{i\in\cI_j}
\]
Then, given any $\bldc \in \rrr^n$, we may say that parity-check $j\in\cJ$ is satisfied by $\bldc$ if and only if
\begin{equation}
\bldx_j(\bldc) \in \code_j \; ,
\label{eq:word_in_SPC_code}
\end{equation}
since~(\ref{eq:parity_check_satisfied}) and~(\ref{eq:word_in_SPC_code}) are equivalent. 
Also, it is easily seen that $\bldc \in \code$ if and only if all parity-checks $j\in\cJ$ are satisfied by $\bldc$. In this case we say that $\bldc$ is a \emph{codeword} of $\code$.

We shall take an example which shall be used to illustrate concepts throughout this paper. Consider the $[4,2]$ linear code over~$\rrr=\mathbb{Z}_3$ with parity-check matrix
\begin{equation}
\cH = \left( \begin{array}{cccc}
1 & 2 & 2 & 1 \\
2 & 0 & 1 & 2
\end{array} \right)
\label{eq:example_PCM}
\end{equation}
Here $\cI_1=\{ 1, 2, 3, 4 \}$, $\cI_2=\{ 1, 3, 4 \}$, and the two single parity-check codes $\code_1$ and $\code_2$, of length $d_1=4$ and $d_2=3$ respectively, are given by
\[
\code_1 = \left\{ ( b_1 \; b_2 \; b_3 \; b_4 ) \; : \; b_1 + 2 \, b_2 + 2 \, b_3 + b_4 = 0 \right\}
\]
and
\[
\code_2 = \left\{ ( b_1 \; b_3 \; b_4 ) \; : \; 2 \, b_1 + b_3 + 2 \, b_4 = 0 \right\} \; .
\]

\section{Decoding as a Linear-Programming Problem} 
\label{sec:lp}
 
Assume that the codeword $\bar{\bldc} = (\bar{c}_1, \bar{c}_2, \cdots, \bar{c}_n) 
\in \code$ has been transmitted over a $q$-ary input
memoryless channel,
and a corrupted word $\bldy = (y_1, y_2, \cdots, y_n) \in \Sigma^n$ has been received. Here $\Sigma$ denotes the set of channel output symbols; we assume that this set either has finite cardinality, or is equal to $\mathbb{R}^l$ or $\mathbb{C}^l$ for some integer $l \ge 1$. In practice, this channel may represent the combination of modulator and physical channel.
We assume hereafter that all information words are equally probable, and so all codewords are 
transmitted with equal probability. 

It was suggested in~\cite{Feldman-thesis} to represent each symbol as a binary vector of 
length $|\rrrm|$, where the entries in the vector are indicators of a symbol 
taking on a particular value. Below, we show how this representation may lead to a generalization of the framework of~\cite{Feldman} to the case of nonbinary coding. This generalization is nontrivial since, while such a representation converts the nonbinary code into a binary code, this binary code is not linear and therefore the analysis in~\cite{Feldman-thesis},~\cite{Feldman} is not directly applicable. 

For use in the following derivation, we shall define the mapping 
\[
\bldxi \; : \; \rrr \longrightarrow \{ 0, 1 \}^{q-1} \subset \mathbb{R}^{q-1} \; , 
\]
by
\[
\bldxi (\alpha) = \bldx = ( x^{(\gamma)} )_{\gamma \in \rrrm } \; , 
\]
such that, for each $\gamma \in \rrrm$,
\[
x^{(\gamma)}=\left\{ \begin{array}{cc}
1 & \textrm{ if } \gamma = \alpha \\
0 & \textrm{ otherwise. }\end{array}\right. \;
\]
We note that the mapping $\bldxi$ is one-to-one, 
and its image is the set of binary vectors of length $q-1$ with Hamming weight 0 or 1. Building on this, we also define
\[
\bldXi \; : \; \rrr^n \longrightarrow \{ 0, 1 \}^{(q-1)n} \subset \mathbb{R}^{(q-1)n} \; , 
\]
according to
\[
\bldXi(\bldc) = ( \bldxi(c_1) \; | \; \bldxi(c_2) \; | \; \cdots \; | \; \bldxi(c_n) ) \; .
\]
We note that $\bldXi$ is also one-to-one.

Now, for vectors $\bldf \in \mathbb{R}^{(q-1)n}$, we adopt the notation
\[
\bldf = ( \bldf_1 \; | \; \bldf_2 \; | \; \cdots \; | \; \bldf_n ) \; , 
\]
where
\[
\forall i \in \cI , \; \bldf_i = ( f_i^{(\alpha)} )_{\alpha \in \rrrm} \; .
\]
Also, we may use this notation to write the inverse of $\bldXi$ as
\[
\bldXi^{-1} (\bldf) = ( \bldxi^{-1}(\bldf_1), \bldxi^{-1}(\bldf_2), \cdots, \bldxi^{-1}(\bldf_n) ) \; .
\]

We also define a function $\bldlambda: \Sigma \longrightarrow \left( {\mathbb R} \cup \{\pm\infty\} \right) ^ {q-1}$ by
\[
\boldsymbol{\lambda} = ( \lambda^{(\alpha)} )_{\alpha \in \rrrm}  \; , 
\]
where, for each $y \in \Sigma$, $\alpha \in \rrrm$,
\[
\lambda^{(\alpha)}(y) = \log \left( \frac{\Prob ( y | 0 )}{ \Prob ( y| \alpha ) } \right) \; , 
\]
and $p(y|c)$ denotes the channel output probability (density) conditioned on the channel input. 
We extend this to a map on $\Sigma^n$ 
by defining $\boldsymbol{\Lambda}(\boldsymbol{y})= (\boldsymbol{\lambda}(y_1) \;|\; \boldsymbol{\lambda}(y_2) \;|\; 
\ldots \;|\; \boldsymbol{\lambda}(y_n))$.

The codeword-error-rate-optimum receiver operates according to the \emph{maximum a posteriori} (MAP) decision rule:
\begin{eqnarray*}
\hat{\bldc} & = & \arg \max_{\bldcc \in \code} \Prob ( \; \bldc \; | \; \bldy \; ) \\
 & = & \arg \max_{\bldcc \in \code} \frac{\Prob ( \; \bldy \; | \; \bldc \; ) \Prob ( \; \bldc \; )}{\Prob ( \; \bldy \; )} \; .
\end{eqnarray*}
Here $\Prob\left(\cdot\right)$ denotes probability if $\Sigma$ has finite cardinality, and probability density if $\Sigma$ has infinite cardinality.

By assumption, the \emph{a priori} probability $\Prob ( \bldc )$ is uniform over codewords, and $\Prob ( \bldy )$ is independent of $\bldc$. Therefore, the decision rule reduces to maximum-likelihood (ML) decoding:
\begin{eqnarray*}
 \hat{\bldc} & = & \arg \max_{\bldcc \in \code} \Prob ( \; \bldy \; | \; \bldc \; ) \\
 & = & \arg \max_{\bldcc \in \code} \prod_{i=1}^n \Prob ( y_i | c_i ) \\
 & = & \arg \max_{\bldcc \in \code} \sum_{i=1}^n \log ( \Prob ( y_i | c_i ) ) \\
 & = & \arg \min_{\bldcc \in \code} \sum_{i=1}^n \log \left( \frac{\Prob ( y_i | 0 )}{\Prob ( y_i | c_i )} \right) \\
 & = & \arg \min_{\bldcc \in \code} \sum_{i=1}^n \boldsymbol{\lambda} (y_i) \bldxi (c_i) ^T \\
 & = & \arg \min_{\bldcc \in \code} \boldsymbol{\Lambda} (\bldy) \bldXi (\bldc) ^T \; ,
\end{eqnarray*}
where we have made use of the memoryless property of the channel, and of the fact that if 
$c_i = \alpha\in\rrrm$, then $\boldsymbol{\lambda} (y_i) \bldxi (c_i)^T = \lambda^{(\alpha)} (y_i)$. This is then equivalent to
\[
\hat{\bldc} = \bldXi^{-1} (\hat{\bldf})
\]
where
\begin{equation}
\hat{\bldf}= \arg \min_{\bldff \in \cK(\code)} \boldsymbol{\Lambda} (\bldy) \bldf^T \; , 
\label{eq:object-function}
\end{equation}
and $\cK(\code)$ represents the convex hull of all points $\bldf\in\mathbb{R}^{(q-1)n}$ which correspond to codewords, i.e.
\[ 
\cK(\code) = H_\mathrm{conv} \big\{ \bldXi{(\bldc)} \; : 
\; \bldc\in\code \big\} \; . 
\]
Therefore it is seen that the ML decoding problem reduces to the minimization of a linear objective function (or cost function) over a polytope in $\rr^{(q-1)n}$.
The number of variables and constraints for this linear program is exponential in $n$, and it is therefore too complex for practical implementation. 
To circumvent this problem, we formulate a relaxed LP problem, as shown next. 

The solution we seek for $\bldf$ (i.e. the desired LP output) is
\begin{equation}
\bldf = \bldXi(\bar{\bldc}) = ( \bldxi(\bar{c}_1) \; | \; \bldxi(\bar{c}_2)  \; | \; \ldots  \; | \; \bldxi(\bar{c}_n) ) \; .
\label{eq:f_desired_solution}
\end{equation}
Note that (\ref{eq:f_desired_solution}) implies that the solution we seek for each $\bldf_i$ ($i \in \cI$) is an indicator function which ``points'' to the $i$-th transmitted symbol $\bar{c}_i$, i.e.
\begin{eqnarray*}
\forall i \in \cI & : & f_i^{(\alpha)} = 
\left\{ \begin{array}{cl}
1 & \textrm{if } \alpha = \bar{c}_i \\
0 & \textrm{otherwise. }\end{array}\right.  
\end{eqnarray*}
We now introduce auxiliary variables whose constraints, along with those of the elements of $\bldf$,
will form the relaxed LP problem. We denote these auxiliary variables by
\[
w_{j,\bldbb} \; \mbox{ for } \;  j\in\cJ, \bldb \in \code_j \; ,
\]
and we denote the vector containing these variables as 
\[
\bldw = \big( \bldw_j \big)_{j \in \cJ} \mbox{ where } \; \bldw_j = \big( w_{j,\bldbb} \big)_{\bldbb \in \code_j} \quad \forall j \in \cJ \; .
\]
The solution we seek for these variables is 
\begin{eqnarray}
\forall j \in \cJ & : & w_{j,\bldbb} = 
\left\{ \begin{array}{cl}
1 & \textrm{if } \bldb = \bldx_j(\bar{\bldc}) \\
0 & \textrm{otherwise. }\end{array}\right. 
\label{eq:w_desired_solution}
\end{eqnarray}
Note that the solution we seek for each $\bldw_j$ ($j \in \cJ$) is an indicator function which ``points'' to the $j$-th transmitted local codeword $\bldx_j(\bar{\bldc})$. Based on~(\ref{eq:w_desired_solution}), we impose the constraints
\begin{eqnarray}
\forall j \in \cJ, \; \forall \bldb \in \code_j,  \quad  w_{j,\bldbb} \ge 0 \; ,
\label{eq:equation-polytope-3} 
\end{eqnarray} 
and
\begin{equation}
\forall j \in \cJ, \quad \sum_{\bldbb \in \code_j} w_{j,\bldbb} = 1 \; .
\label{eq:equation-polytope-4} 
\end{equation} 

Finally, we note that the solution we seek (given by the combination of~(\ref{eq:f_desired_solution}) and~(\ref{eq:w_desired_solution})) satisfies the further constraints 
\begin{eqnarray}
&  \forall j \in \cJ, \; \forall i \in \cI_j, \; \forall \alpha \in \rrrm, \nonumber \\
& f_i^{(\alpha)} = 
\sum_{\bldbb \in \code_j, \; b_i=\alpha} w_{j,\bldbb} \; .
\label{eq:equation-polytope-5} 
\end{eqnarray} 
It is interesting to note that from~(\ref{eq:equation-polytope-3}) and~(\ref{eq:equation-polytope-4}), each vector $\bldw_j$ (for $j \in \cJ$) may be interpreted as a probability distribution for the local codeword $\bldb \in \code_j$, in which case each $\bldf_i$ (for $i \in \cI$) has a natural interpretation (via (\ref{eq:equation-polytope-5})) as the corresponding probability distribution for the $i$-th coded symbol $c_i \in \rrr$. The following example illustrates the connection~(\ref{eq:equation-polytope-5}) between $\bldf$ and $\bldw$.

\begin{example}
Consider the example $[4,2]$ code over $\mathbb{Z}_3$ defined by the parity-check matrix~(\ref{eq:example_PCM}). The second row $\cH_2$ of the parity-check matrix corresponds to the parity-check equation 
\[
2 b_1 + b_3 + 2 b_4 = 0 \; 
\] 
over $\mathbb{Z}_3$. Here $\bldb = (b_1 \;\; b_3 \;\; b_4) \in \code_2$. Assume that the values of $w_{2, \bldbb}$ for $\bldb \in \code_2$ are as given in the following table.
\medskip
\begin{center}
\begin{tabular}{l@{\hspace{0.4cm}}p{0.9cm}|p{1.1cm}p{0.9cm}|p{0.9cm}@{\hspace{0.4cm}}l}
$b_1 b_3 b_4$ & $w_{2, \bldbb}$ & $b_1 b_3 b_4$ & $w_{2, \bldbb}$ & $b_1 b_3 b_4$ & $w_{2, \bldbb}$ \\
\hline
$\;\; 000$ & $0.01$ & $\;\; 102$ & $0.05$ & $\;\; 201$ & $0.15$\\
$\;\; 011$ & $0.04$ & $\;\; 110$ & $0.07$ & $\;\; 212$ & $0.32$\\
$\;\; 022$ & $0.05$ & $\;\; 121$ & $0.08$ & $\;\; 220$ & $0.23$
\end{tabular}
\end{center}
\medskip

Then, some of the values of $f_i^{(\alpha)}$ are as follows:
\begin{eqnarray*}
f_1^{(2)} & = & 0.15 + 0.32 + 0.23 = 0.7 \; ; \\
f_2^{(1)} & = & 0.04 + 0.07 + 0.32 = 0.43 \; ; \\
f_3^{(2)} & = & 0.05 + 0.05 + 0.32 = 0.42 \; .  
\end{eqnarray*}
\medskip
\end{example}

Constraints~(\ref{eq:equation-polytope-3})-(\ref{eq:equation-polytope-5}) may be interpreted as the statement that for all $j \in \cJ$, the vector $\hat{\bldf}_j = ( \bldf_i )_{i \in \cI_j}$ lies in the convex hull $\cK(\code_j)$. 
Constraints~(\ref{eq:equation-polytope-3})-(\ref{eq:equation-polytope-5}) form a polytope which we denote by $\cQ$. The minimization of
the objective function~(\ref{eq:object-function}) over $\cQ$ forms the relaxed LP decoding problem.
This LP is defined by $O(qn + q^d m)$ variables and $O(qn + q^d m)$ constraints, and
therefore, the number of variables and of constraints scales as approximately $q^d$. 

We note that the further constraints
\begin{equation}
\forall j \in \cJ, \; \forall \bldb \in \code_j,  \quad  w_{j,\bldbb} \le 1 \; ,
\label{eq:equation-polytope-0} 
\end{equation}  
\begin{equation}
\forall i \in \cI, \; \forall \alpha\in\rrrm, \quad 0 \le f_i^{(\alpha)} \le 1 \; .  
\label{eq:equation-polytope-1} 
\end{equation} 
and
\begin{equation}
\forall i \in \cI, \quad \sum_{\alpha\in\rrrm} f_i^{(\alpha)} \le 1 \; . 
\label{eq:equation-polytope-2} 
\end{equation} 
follow from the constraints~(\ref{eq:equation-polytope-3})-(\ref{eq:equation-polytope-5}), for any 
$(\bldf, \bldw) \in \cQ$.

Now we may define the decoding algorithm, which works as follows. 
The decoder solves the LP problem of minimizing the objective 
function~(\ref{eq:object-function}) subject to the constraints~(\ref{eq:equation-polytope-3})-(\ref{eq:equation-polytope-5}).  
If $\bldf \in \{0,1\}^{(q-1)n}$, the output is the codeword $\bldXi^{-1}(\bldf)$ (we shall prove in the next section that 
this output is indeed a codeword). This codeword may then be the correct one (we call this `correct decoding') or an incorrect one (we call this `incorrect decoding'). If $\bldf \notin \{0,1\}^{(q-1)n}$, the decoder reports a `decoding failure'. Note that in this paper, we say that the decoder makes a \emph{codeword error} when the decoder output is not equal to the transmitted codeword (this could correspond to a `decoding failure', or to an `incorrect decoding').

The time complexity of an LP solver depends on the number of variables and constraints in 
the LP problem. The \emph{simplex method} is a popular and practically efficient algorithm for solving LP problems. 
However, its worst-case time complexity has been shown to be exponential in the number of variables. There are other
known LP solvers, such as solvers that are based on interior-point methods~\cite[Chapter 11]{Boyd}, which have time complexity polynomial in the number of variables and constraints. 
For more detail the reader may also refer to~\cite{Schrijver}. We note, however, that the standard iterative decoding algorithms (such as the min-sum or sum-product algorithms) have time complexity which is linear in the block length of the code, and therefore significantly outperform the LP decoder in terms of efficiency.


\section{Polytope Properties} 
\label{sec:pol-properties}

The analysis in this section is a direct generalization of the results in~\cite{Feldman}. 

\begin{definition}
An \emph{integral point} in a polytope is a point with all \emph{integer} coordinates.  
\end{definition}

\begin{proposition}
$\,$
\begin{itemize} 
\item[1)]
Let $(\bldf, \bldw) \in \cQ$, and $f_i^{(\alpha)} \in \{ 0,1 \}$ for every $i \in \cI$, $\alpha\in\rrrm$. 
Then $\bldXi^{-1}(\bldf) \in \code$.
\item[2)]
Conversely, for every codeword $\bldc \in \code$, there exists 
$\bldw$ such that $(\bldf, \bldw)$ is an integral point in $\cQ$ with 
$\bldf = \bldXi(\bldc)$.
\end{itemize}
\end{proposition} 
 
\begin{proof}  
\begin{enumerate}
\item
Suppose $(\bldf, \bldw) \in \cQ$, and 
$f_i^{(\alpha)} \in \{ 0,1 \}$ for every $i \in \cI$, $\alpha\in\rrrm$.

Let $\bldc = \bldXi^{-1}(\bldf)$; by~(\ref{eq:equation-polytope-2}), this is well defined. Now, fix some $j\in\cJ$ and define $\bldt = \bldx_j(\bldc)$. Note that from these definitions it follows that for any $i\in\cI$, $\alpha\in\rrrm$, $f_i^{(\alpha)} = 1$ if and only if $t_i=\alpha$. Now let $\bldr \in \code_j$, $\bldr \ne \bldt$. Since $\bldr$ and $\bldt$ are distinct, there must exist $\alpha\in\rrrm$ and $l\in\cI_j$ such that either $r_l=\alpha$ and $t_l \ne \alpha$, or $t_l=\alpha$ and $r_l \ne \alpha$. We examine these two cases separately.
\begin{itemize}
\item
If $r_l=\alpha$ and $t_l \ne \alpha$, then by~(\ref{eq:equation-polytope-5})
\[
f_l^{(\alpha)} = 0 = \sum_{\bldbb \in \code_j , \; b_l=\alpha} w_{j, \bldbb} \; .
\]
Therefore $w_{j,\bldbb} = 0$ for all $\bldb \in \code_j$ with $b_l=\alpha$,
and in particular $w_{j,\bldrr} = 0$. 
\item
If $t_l=\alpha$ and $r_l \ne \alpha$, then by~(\ref{eq:equation-polytope-4}) and~(\ref{eq:equation-polytope-5})
\begin{eqnarray*}
0 & = & 1 - f_{l}^{(\alpha)} \\
& = & \sum_{\bldbb \in \code_j} w_{j,\bldbb} - \sum_{\bldbb \in \code_j, \; b_l=\alpha} w_{j,\bldbb} \\
& = & \sum_{\bldbb \in \code_j, \; b_l \ne \alpha} w_{j,\bldbb} \; .
\end{eqnarray*}
Therefore $w_{j,\bldbb} = 0$ for all $\bldb \in \code_j$ with $b_l \ne \alpha$,
and in particular $w_{j,\bldrr} = 0$. 
\end{itemize}

It follows that $w_{j,\bldrr} = 0$ for all $\bldr \in \code_j$, $\bldr \neq
\bldt$. But by~(\ref{eq:equation-polytope-4}) this implies that $\bldt
\in \code_j$ (and that $w_{j,\bldtt} = 1$). Applying this argument for every $j\in\cJ$ implies
$\bldc \in \code$.
\item
For $\bldc \in \code$, we let $\bldf = \bldXi(\bldc)$. For each parity-check $j\in\cJ$, we let $\bldt = \bldx_j(\bldc) \in \code_j$ and then set
\begin{eqnarray*}
\forall j \in \cJ : && w_{j,\bldbb} = 
\left\{ \begin{array}{cl}
1 & \textrm{if } \bldb = \bldt \\
0 & \textrm{otherwise. }\end{array}\right.
\end{eqnarray*}
It is easily checked that the resulting point $(\bldf, \bldw)$ is integral and satisfies
constraints~(\ref{eq:equation-polytope-3})-(\ref{eq:equation-polytope-5}).
\end{enumerate}
\end{proof}

The following proposition assures the so-called \emph{ML certificate} property. 
\begin{proposition}
Suppose that the decoder outputs a codeword $\bldc \in \code$. Then, $\bldc$ 
is the maximum-likelihood codeword. 
\end{proposition} 

The proof of this proposition is straightforward. The reader can refer to a similar proof for 
the binary case in~\cite{Feldman}. 

 
\section{Codeword-Independent Decoder Performance\label{sec:codeword_independence}} 

In this section, we state and prove a theorem on decoder performance, namely, that under a certain symmetry condition, the probability of codeword error is independent of the transmitted codeword. The proof generalizes the corresponding proof for the binary case which may be found in~\cite{Feldman-thesis,Feldman}.

{\bf Symmetry Condition.}  

For each $\beta\in\rrr$, there exists a bijection 
\[
\tau_{\beta} \; : \; \Sigma \longrightarrow \Sigma \; , 
\]
such that the channel output probability (density) conditioned on the channel input satisfies
\begin{equation}
p(y | \alpha) = p(\tau_{\beta}(y)|\alpha - \beta) \; ,
\label{eq:symmetry_condition}
\end{equation}
$\forall y\in \Sigma$, $\alpha \in \rrr$.
When $\Sigma$ is equal to $\mathbb{R}^l$ or $\mathbb{C}^l$ for $l \ge 1$, the mapping $\tau_{\beta}$ 
is assumed to be isometric with respect to Euclidean distance in $\Sigma$, for every $\beta\in\rrr$.

\vspace{3mm}
Note that the symmetry condition above is very similar to that introduced in \cite{Hof_Sason_Shamai} which guarantees codeword-independent performance under ML decoding. 
\begin{theorem}
Under the stated symmetry condition, the probability of codeword error is independent of the transmitted codeword.
\label{htrm:equl-prob} 
\end{theorem}

\begin{proof}
We shall prove the theorem for the case where $\Sigma$ has infinite cardinality; the case of discrete $\Sigma$ may be handled similarly. Fix some codeword $\bldc \in \code$, $\bldc \ne \zeros$. We wish to prove that
\[
\mbox{Pr}(\mbox{Err} \; | \; \bldc) = \mbox{Pr}(\mbox{Err} \; | \; \zeros) \; , 
\]
where $\mbox{Pr}(\mbox{Err} \; | \; \bldc)$ denotes the probability of codeword error given that the codeword $\bldc$ was transmitted. 

Now
\[
\mbox{Pr}(\mbox{Err} \; | \; \bldc) = \mbox{Pr}(\bldy \in B(\bldc) \; | \; \bldc) \; , 
\]
where
\begin{equation*}
\begin{split}
B(\bldc) = \{ \bldy \in \Sigma^n \; : \; \exists (\bldf, \bldw) & \in \cQ, \bldf \ne \bldXi(\bldc) \\
 & \mbox{with } \boldsymbol{\Lambda}(\bldy) \bldf ^T \le \boldsymbol{\Lambda} (\bldy) \bldXi(\bldc) ^T \} \; .
\end{split}
\end{equation*}
Here $B(\bldc)$ is the set of all received words which may cause codeword error, given that $\bldc$ was transmitted. Recall that the elements of $\boldsymbol{\Lambda} (\bldy)$ are given by
\begin{equation}
\lambda^{(\alpha)} (y_i) = \log \left( \frac{\Prob ( y_i | 0 )}{ \Prob ( y_i | \alpha ) } \right) \; ,
\label{eq:lambda_def}
\end{equation}
for $i\in\cI$, $\alpha\in\rrrm$. Also
\[
\mbox{Pr}(\mbox{Err} \; | \; \zeros) = \mbox{Pr}(\bldy \in B(\zeros) \; | \; \zeros)
\]
where
\begin{equation*}
\begin{split}
B(\zeros) = \{ \tilde{\bldy} \in \Sigma^n \; : \; \exists (\tilde{\bldf}, \tilde{\bldw}) & \in \cQ, \tilde{\bldf} \ne \bldXi(\zeros) \\
 & \mbox{with } \boldsymbol{\Lambda}(\tilde{\bldy}) \tilde{\bldf} ^T \le \boldsymbol{\Lambda}(\tilde{\bldy}) \bldXi(\zeros) ^T \} \; .
\end{split}
\end{equation*}

So we write
\begin{equation}
\mbox{Pr}(\mbox{Err} \; | \; \bldc) = \int_{\bldyy \in B(\bldcc)} \Prob ( \; \bldy \; | \; \bldc \; ) \; d\bldy
\label{eq:integral_1}
\end{equation}
and
\begin{equation}
\mbox{Pr}(\mbox{Err} \; | \; \zeros) = \int_{\tilde{\bldyy} \in B(\zeross)} \Prob ( \; \tilde{\bldy} \; | \; \zeros \; ) \; d\tilde{\bldy} \; .
\label{eq:integral_2}
\end{equation}
Now, setting $\alpha=\beta$ in the symmetry condition~(\ref{eq:symmetry_condition}) yields
\begin{equation}
p(y|\beta) = p(\tau_{\beta}(y)|0)
\label{eq:equality_in_symmetry_condition}
\end{equation}
for any $y \in \Sigma$, $\beta \in \rrr$.

We now define $\bldG : \Sigma^n \longrightarrow \Sigma^n$ and $\tilde{\bldy}$ as follows. 
\[
\tilde{\bldy} = \bldG(\bldy) \quad \mbox{ s.t. } \quad \forall i \in \cI: \;  \tilde{y}_i = \tau_{\beta}(y_i)
\mbox{ where } \beta = c_i \; . 
\]
We note that $\bldG$ is a bijection from the set $\Sigma^n$ to itself, and that if $\bldy, \bldz \in \Sigma^n$ and 
$\beta = c_i$ then 
\[
\| y_i - z_i \| ^2 = \| \tau_{\beta}(y_i) - \tau_{\beta}(z_i) \| ^2
\] 
and so
\[
\| \bldG(\bldy) - \bldG(\bldz) \| ^2 = \| \bldy - \bldz \| ^2
\] 
i.e. $\bldG$ is isometric with respect to Euclidean distance in $\Sigma^n$.

We prove that the integral~(\ref{eq:integral_1}) may be transformed to~(\ref{eq:integral_2}) via the substitution $\tilde{\bldy} = \bldG(\bldy)$. First, we have
\begin{eqnarray*}
\Prob ( \; \bldy \; | \; \bldc \; ) & = & \prod_{i\in\cI} \Prob ( y_i | c_i ) \\
 & = & \prod_{\beta\in\rrr} \prod_{i\in\cI , c_i=\beta} \Prob ( y_i | \beta ) \\
 & = & \prod_{\beta\in\rrr} \prod_{i\in\cI, c_i=\beta} \Prob ( \tau_{\beta}(y_i) | 0 ) \\
 & = & \prod_{\beta\in\rrr} \prod_{i\in\cI, c_i=\beta} \Prob ( \tilde{y}_i | 0 ) \\
 & = & \prod_{i\in\cI} \Prob ( \tilde{y}_i | 0 ) \\
 & = & \Prob ( \; \tilde{\bldy} \; | \; \zeros \; ) \; .
\end{eqnarray*}
Since $\bldG$ is isometric with respect to Euclidean distance in $\Sigma^n$, it follows that the Jacobian determinant of the transformation is equal to unity. Therefore, to complete the proof, we need only show that 
\[
\bldy \in B(\bldc) \mbox{ if and only if } \tilde{\bldy} \in B(\zeros) \; .
\]

We begin by relating the elements of $\boldsymbol{\Lambda} (\bldy)$ to the elements of 
$\boldsymbol{\Lambda}(\tilde{\bldy})$. 
Let $i\in\cI$, $\alpha\in\rrrm$. Suppose $c_i=\beta \in \rrr$. We then have
\begin{eqnarray*}
\lambda^{(\alpha)} (y_i) & = & \log \left( \frac{\Prob ( y_i | 0 )}{ \Prob ( y_i | \alpha ) } \right) \\
 & = & \log \left( \frac{\Prob ( \tau_{\beta}(y_i) | -\beta )}{ \Prob ( \tau_{\beta}(y_i) | \alpha-\beta ) } \right) \\
 & = & \log \left( \frac{\Prob ( \tilde{y}_i | -\beta )}{ \Prob ( \tilde{y}_i | \alpha-\beta ) } \right) \; .
\end{eqnarray*}
This yields
\[
\lambda^{(\alpha)}(y_i)  = \left\{ \begin{array}{ccc}
\lambda^{(\alpha)}(\tilde{y_i}) & \textrm{ if } \beta=0 \\
-\lambda^{(-\alpha)}(\tilde{y_i}) & \textrm{ if } \alpha=\beta \\
\lambda^{(\alpha-\beta)} (\tilde{y_i}) - \lambda^{(-\beta)} (\tilde{y_i}) 
& \textrm{ otherwise. }\end{array}\right.
\]

Next, for any point $(\bldf,\bldw)\in\cQ$ we define a new point $(\tilde{\bldf},\tilde{\bldw})$ as follows. For $\beta = c_i$ and all $i\in\cI$, $\alpha \in \rrrm$, 
\begin{equation}
\tilde{f}_i^{(\alpha)} = \left\{ \begin{array}{ccc}
1 - \sum_{\gamma\in\rrrm} f_i^{(\gamma)} & \textrm{ if } \alpha=-\beta \\
f_i^{(\alpha+\beta)} & \textrm{ otherwise. }\end{array}\right. 
\label{eq:f_to_ftilde}
\end{equation}
For all $j\in\cJ$, $\bldr \in \code_j$ we define
\[
\tilde{w}_{j,\bldrr} = w_{j,\bldbb}
\]
where
\[
\bldb = \bldr + \bldx_j(\bldc) \; .
\]

Next we prove that for every $(\bldf,\bldw) \in \cQ$, the new point $(\tilde{\bldf},\tilde{\bldw})$ lies in $\cQ$ and thus is a feasible solution for the LP. Constraints~(\ref{eq:equation-polytope-3}) and~(\ref{eq:equation-polytope-4}) obviously hold from the definition of $\tilde{\bldw}$. To verify~(\ref{eq:equation-polytope-5}), we let $j\in\cJ$, $i\in\cI_j$ and $\alpha \in \rrrm$. We also let $\beta = c_i$. We now check two cases: 
\begin{itemize}
\item
If $\alpha=-\beta$,
\begin{eqnarray*}
\tilde{f}_i^{(\alpha)} & = & 1 - \sum_{\gamma\in\rrrm} f_i^{(\gamma)} \\
 & = & \sum_{\bldbb \in \code_j} w_{j,\bldbb} - \sum_{\gamma\in\rrrm} \sum_{\bldbb \in \code_j, \; b_i=\gamma} w_{j,\bldbb} \\
 & = & \sum_{\bldbb \in \code_j, \; b_i=0} w_{j,\bldbb} \\
 & = & \sum_{\bldrr \in \code_j, \; r_i=\alpha} \tilde{w}_{j,\bldrr} \; .
\end{eqnarray*}
\item
If $\alpha \ne -\beta$,
\begin{eqnarray*}
\tilde{f}_i^{(\alpha)} = f_i^{(\alpha+\beta)} & = & \sum_{\bldbb \in \code_j, \; b_i=\alpha+\beta} w_{j,\bldbb} \\
 & = & \sum_{\bldrr \in \code_j, \; r_i=\alpha} \tilde{w}_{j,\bldrr} \; .
\end{eqnarray*}
\end{itemize}
Therefore $(\tilde{\bldf},\tilde{\bldw}) \in \cQ$, i.e. $(\tilde{\bldf},\tilde{\bldw})$ is a feasible solution for the LP. We write $(\tilde{\bldf},\tilde{\bldw}) = \bldL(\bldf,\bldw)$. We also note that the mapping $\bldL$ is a bijection from $\cQ$ to itself; this is easily shown by verifying the inverse
\begin{equation}
f_i^{(\alpha)} = \left\{ \begin{array}{ccc}
1 - \sum_{\gamma\in\rrrm} \tilde{f}_i^{(\gamma)} & \textrm{ if } \alpha=\beta \\
\tilde{f}_i^{(\alpha-\beta)} & \textrm{ otherwise }\end{array}\right. 
\label{eq:ftilde_to_f}
\end{equation}
for all $i\in\cI$, $\alpha \in \rrrm$, and 
\[
w_{j,\bldbb} = \tilde{w}_{j,\bldrr} 
\]
where
\[
\bldr = \bldb - \bldx_j(\bldc)
\]
for all $j\in\cJ$, $\bldb \in \code_j$.

We now prove that for every $(\bldf,\bldw) \in \cQ$, $(\tilde{\bldf},\tilde{\bldw}) = \bldL(\bldf,\bldw)$ satisfies
\begin{equation}
\boldsymbol{\Lambda} (\bldy) \bldf ^T - \boldsymbol{\Lambda}(\bldy) \bldXi(\bldc) ^T 
= \boldsymbol{\Lambda}(\tilde{\bldy}) \tilde{\bldf} ^T - \boldsymbol{\Lambda}(\tilde{\bldy}) \bldXi(\zeros) ^T \; .
\label{eq:relative_cost_fn_1}
\end{equation}
We achieve this by proving
\begin{equation}
\boldsymbol{\lambda} (y_i) \bldf_i ^T - \boldsymbol{\lambda} (y_i) \bldxi(c_i) ^T = \boldsymbol{\lambda}(\tilde{y_i}) \tilde{\bldf}_i ^T - \boldsymbol{\lambda} (\tilde{y_i}) \bldxi(0) ^T  
\label{eq:relative_cost_fn_2}
\end{equation}
for every $i\in \cI$. We may then obtain~(\ref{eq:relative_cost_fn_1}) by summing~(\ref{eq:relative_cost_fn_2}) over $i\in\cI$. Let $\beta = c_i$. We consider two cases:
\begin{itemize}
\item
If $\beta=0$,~(\ref{eq:relative_cost_fn_2}) becomes
\[
\boldsymbol{\lambda} (y_i) \bldf_i ^T = \boldsymbol{\lambda}(\tilde{y_i}) \tilde{\bldf}_i ^T 
\]
which holds since $\lambda^{(\alpha)}(\tilde{y_i}) = \lambda^{(\alpha)}(y_i)$ and $\tilde{f}_i^{(\alpha)} = f_i^{(\alpha)}$ for all $\alpha\in\rrrm$ in this case.
\item
If $\beta \ne 0$,
\begin{equation*}
\begin{split}
& \boldsymbol{\lambda} (y_i) \bldf_i ^T - \boldsymbol{\lambda} (y_i) \bldxi(c_i) ^T \\
& = \sum_{\gamma\in\rrrm} \lambda^{(\gamma)} (y_i) f_i^{(\gamma)} - \lambda^{(\beta)} (y_i) \\
 & = \sum_{\substack{\gamma\in\rrrm \\ \gamma \ne \beta}} \left( \lambda^{(\gamma-\beta)} (\tilde{y_i}) 
 - \lambda^{(-\beta)} (\tilde{y_i}) \right)  f_i^{(\gamma)} - \lambda^{(-\beta)} (\tilde{y_i}) f_i^{(\beta)} 
 + \lambda^{(-\beta)} (\tilde{y_i}) \\
 & = \sum_{\substack{\alpha\in\rrrm \\ \alpha \ne -\beta}} \lambda^{(\alpha)} (\tilde{y_i}) f_i^{(\alpha+\beta)} 
 + \lambda^{(-\beta)} (\tilde{y_i}) \left( 1 - \sum_{\gamma\in\rrrm} f_i^{(\gamma)} \right) \\
 & = \sum_{\alpha\in\rrrm} \lambda^{(\alpha)} (\tilde{y_i}) \tilde{f}_i^{(\alpha)} \\
 & = \boldsymbol{\lambda} (\tilde{y_i}) \tilde{\bldf}_i ^T - \boldsymbol{\lambda} (\tilde{y_i}) \bldxi(0)^T
\end{split}
\end{equation*}
\end{itemize}
where we have made use of the substitution $\alpha = \gamma - \beta$ in the third line. Therefore~(\ref{eq:relative_cost_fn_2}) holds, proving~(\ref{eq:relative_cost_fn_1}).

Finally, we note that it is easy to show, using~(\ref{eq:f_to_ftilde}) and~(\ref{eq:ftilde_to_f}), that $\bldf = \bldXi(\bldc)$ if and only if $\tilde{\bldf} = \bldXi(\zeros)$.
Putting together these results, we may make the following statement. Suppose we are given $\bldy, \tilde{\bldy} \in \Sigma^n$ with $\tilde{\bldy} = \bldG(\bldy)$. Then the point $(\bldf,\bldw) \in \cQ$ satisfies $\bldf \ne \bldXi(\bldc)$ and $\boldsymbol{\Lambda}(\bldy) \bldf ^T \le \boldsymbol{\Lambda} (\bldy) \bldXi(\bldc) ^T$ if and only if the point $(\tilde{\bldf},\tilde{\bldw}) = \bldL(\bldf,\bldw) \in \cQ$ satisfies $\tilde{\bldf} \ne \bldXi(\zeros)$ and $\boldsymbol{\Lambda}(\tilde{\bldy}) \tilde{\bldf} ^T \le \boldsymbol{\Lambda}(\tilde{\bldy}) \bldXi(\zeros) ^T$. This statement, along with the fact that both $\bldG$ and $\bldL$ are bijective, proves that 
\[
\bldy \in B(\bldc) \mbox{ if and only if } \tilde{\bldy} \in B(\zeros) \; .
\]
\end{proof}

We next provide, with details, some examples of modulator-channel combinations for which the symmetry conditions hold.

\begin{example} 
\emph{Discrete memoryless $q$-ary symmetric chan\-nel.}
Here we denote the ring elements by $\rrr = \{ a_0, a_1, \cdots, a_{q-1} \}$. Also $\Sigma = \{ s_0, s_1, \cdots, s_{q-1} \}$, 
where the channel output probability conditioned on the channel input satisfies, 
for each $t,k\in\left\{0,1,\cdots, q-1 \right\}$,
\[
\Prob ( s_t | a_k ) = \left\{ \begin{array}{cc}
(1-p) & \textrm{ if } t = k \\
p/(q-1) & \textrm{ otherwise } 
\end{array}\right. \; , 
\]
where $p$ represents the probability of transmission error.
Here we may define the mapping $\tau_\beta$ for each $\beta\in\rrr$ according to 
\[
\tau_\beta(s_t) = s_{\ell} \textrm{ where } a_{\ell} = a_t - \beta 
\]
for all $t\in\left\{0,1,\cdots q-1 \right\}$. It is easy to check that these mappings are bijective and satisfy the symmetry condition.
\end{example}

\begin{example}
\emph{Orthogonal modulation over AWGN.}
Here $\Sigma = \mathbb{R}^q$, and denoting the ring elements by $\rrr = \{ a_0, a_1, \cdots, a_{q-1} \}$, the modulation mapping may be written without loss of generality as
\[
\cM \; : \; \rrr \longrightarrow \mathbb{R}^q \; , 
\]
such that, for each $k = 0, 1, \cdots, q-1$,
\[
\cM (a_k) = \bldx = ( x^{(0)}, x^{(1)}, \cdots, x^{(q-1)} ) \; , 
\]
where
\[
x^{(t)}=\left\{ \begin{array}{cc}
1 & \textrm{ if } t = k \\ 
0 & \textrm{ otherwise. }\end{array}\right.
\]
Here we may define the mapping $\tau_\beta$ for each $\beta\in\rrr$ according to (where $\bldy  = ( y^{(0)}, y^{(1)}, \cdots, y^{(q-1)} ) \in \mathbb{R}^q$, $\bldz = ( z^{(0)}, z^{(1)}, \cdots, z^{(q-1)} ) \in \mathbb{R}^q)$ 
\[ 
\tau_\beta(\bldy) = \bldz
\]
such that for each $l\in\left\{0,1,\cdots, q-1 \right\}$, 
\[
z^{(\ell)} = y^{(k)} \textrm{ where } a_k = a_l + \beta.
\]
It is easily checked that these mappings are bijective and isometric, and satisfy the symmetry condition.
\end{example}

\begin{example}
\emph{$q$-ary PSK modulation over AWGN.\label{ex:qary_PSK_AWGN}}

Here $\Sigma = \mathbb{C}$, and again denoting the ring elements by $\rrr = \{ a_0, a_1, \cdots, a_{q-1} \}$, the modulation mapping may be written without loss of generality as
\[
M \; : \; \rrr \mapsto \mathbb{C} 
\]
such that
\begin{equation}
M (a_k) = \exp \left( \frac{\imath 2 \pi k}{q} \right) 
\label{eq:modulation_mapping}
\end{equation}
for $k = 0, 1, \cdots, q-1$ (here $\imath=\sqrt{-1}$). Here~(\ref{eq:equality_in_symmetry_condition}), together with the rotational symmetry of the $q$-ary PSK constellation, motivates us to define, for every $\beta= a_k \in\rrr$,
\begin{equation}
\tau_\beta(x) = \exp \left( \frac{-\imath 2 \pi k}{q} \right) \cdot x \qquad \forall x \in \mathbb{C} 
\label{eq:tau_defn_PSK}
\end{equation}
Next, we also impose the condition that $\rrr$ under addition is a cyclic group. To see why we impose this condition, let $\alpha = a_k \in\rrr$ and $\beta = a_l \in\rrr$. By the symmetry condition we must have
\[
p(y_i|\alpha+\beta) = p(\tau_{\alpha+\beta}(y_i)|0) 
\]
and also
\begin{eqnarray*}
p(y_i|\alpha+\beta) & = & p(\tau_{\beta}(y_i)|\alpha) \\
& = & p(\tau_{\alpha}(\tau_{\beta}(y_i))|0) \; .
\end{eqnarray*}
In order to equate these two expressions, we impose the condition $\tau_{\alpha+\beta}(x) = \tau_{\alpha}(\tau_{\beta}(x))$ for all $x \in \mathbb{C}$, $\alpha,\beta \in\rrr$. Letting $\alpha+\beta = a_p\in\rrr$, and using~(\ref{eq:tau_defn_PSK}) yields
\[
\exp \left( \frac{- \imath 2 \pi k}{q} \right) \cdot \exp \left( \frac{- \imath 2 \pi l}{q} \right) = \exp \left( \frac{- \imath 2 \pi p}{q} \right)
\] 
and thus $p\equiv (k+l) \mod q$. 

Therefore, we must have 
\begin{equation}
a_k + a_l = a_{(k+l) \!\!\!\! \mod q}
\label{eq:cyclic_group}
\end{equation} 
for all $a_k, a_l \in \rrr$. This implies that $\rrr$, under addition, is a cyclic group.

It is easy to check that the condition that $\rrr$ under addition is cyclic, encapsulated by~(\ref{eq:cyclic_group}), along with the modulation mapping (\ref{eq:modulation_mapping}), satisfies the symmetry condition, where the appropriate mappings $\tau_{\beta}$ are given by~(\ref{eq:tau_defn_PSK}). This means that codeword-independent performance is guaranteed for such systems using nonbinary codes with PSK modulation. This applies to AWGN, flat fading wireless channels, and OFDM systems transmitting over frequency selective channels with sufficiently long cyclic prefix.

\end{example}

\section{Linear Programming Pseudocodewords} 
\label{sec:lp-pseudocodewords}
\begin{definition}
A \emph{linear-programming pseudocodeword} (LP pseudocodeword) of the code $\code$, with parity-check matrix $\cH$, is a pair $(\bldh, \bldz)$ where $\bldh \in \mathbb{R}^{(q-1)n}$ and 
\[
\bldz = \big( \; z_{j,\bldbb} \; \big)_{j \in \cJ, \bldbb \in \code_j} \; ,  
\]
where $z_{j,\bldbb}$ is a nonnegative integer for all $j\in\cJ$, $\bldb \in \code_j$, such that the following constraints are satisfied:
\begin{eqnarray}
&  \forall j \in \cJ, \; \forall i \in \cI_j, \; \forall \alpha\in\rrrm, \nonumber \\
& h_i^{(\alpha)} = 
\sum_{\bldbb \in \code_j, \; b_i=\alpha} z_{j,\bldbb} \; ,
\label{eq:LP_PCW_1} 
\end{eqnarray} 
and
\begin{equation}
\forall j \in \cJ, \quad \sum_{\bldbb \in \code_j} z_{j,\bldbb} = M \; ,
\label{eq:LP_PCW_2} 
\end{equation} 
where $M$ is a nonnegative integer independent of $j$.
\end{definition}
It follows from~(\ref{eq:LP_PCW_1}) that $h_i^{(\alpha)}$ is a nonnegative integer for all $i\in\cI$, $\alpha\in\rrrm$. We note that the further constraints
\begin{equation}
\forall j \in \cJ, \; \forall \bldb \in \code_j,  \quad  z_{j,\bldbb} \le M \; ,
\label{eq:LP_PCW_3} 
\end{equation}  
\begin{equation}
\forall i \in \cI, \; \forall \alpha\in\rrrm, \quad 0 \le h_i^{(\alpha)} \le M \; ,
\label{eq:LP_PCW_4} 
\end{equation} 
and
\begin{equation}
\forall i \in \cI, \quad \sum_{\alpha\in\rrrm} h_i^{(\alpha)} \le M \; ,
\label{eq:LP_PCW_5} 
\end{equation} 
follow from the constraints~(\ref{eq:LP_PCW_1}) and~(\ref{eq:LP_PCW_2}). 

For each $i\in\cI$, we also define
\begin{equation}
h_i^{(0)} = M - \sum_{\alpha\in\rrrm} h_i^{(\alpha)} \; .  
\label{eq:h_sum_equals_M}
\end{equation}
By~(\ref{eq:LP_PCW_5}), $h_i^{(0)}$ is a nonnegative integer for all $i\in\cI$. Now, for any $j\in\cJ$, $i\in\cI_j$ we have 
\begin{eqnarray*}
h_i^{(0)} & = & M - \sum_{\alpha\in\rrrm} h_i^{(\alpha)} \\
 & = & \sum_{\bldbb \in \code_j} z_{j,\bldbb} - \sum_{\alpha\in\rrrm} \sum_{\bldbb \in \code_j, b_i=\alpha} z_{j,\bldbb} \\
 & = & \sum_{\bldbb \in \code_j, b_i=0} z_{j,\bldbb} 
\end{eqnarray*}
where we have used~(\ref{eq:LP_PCW_1}) and~(\ref{eq:LP_PCW_2}).

Corresponding to the LP pseudocodeword $(\bldh, \bldz)$ defined above, we define the \emph{normalized LP 
pseudocodeword} as the vector obtained by scaling of $(\bldh, \bldz)$ by a factor $1/M$. 
We also define the $n \times q$ \emph{LP pseudocodeword matrix} 
\[
\mathsf{H} = \Big( h_i^{(\alpha)} \Big)_{i \in \cI; \, \alpha\in\rrr}  \; .
\] 
The \emph{normalized LP pseudocodeword matrix} is defined as $(1/M) \cdot \mathsf{H}$.

Note that if we interpret $\{ z_{j,\bldbb} / M \}$ (for each $j \in \cJ$) as a probability distribution for the local codeword $\bldb \in \code_j$, then the $i$-th row of the normalized LP pseudocodeword matrix (for $i \in \cI$) can be interpreted as the corresponding probability distribution for the $i$-th coded symbol $c_i \in \rrr$. This idea of interpretating pseudocodewords as probability distributions was used in \cite{FKKR} for the binary case.  
\begin{example}
As an illustration, we provide an LP pseudocodeword for the example $[4,2]$ code over $\mathbb{Z}_3$ defined by the parity-check matrix~(\ref{eq:example_PCM}). The reader may check that 
\begin{equation}
(h^{(1)}_1, h^{(1)}_2, h^{(1)}_3, h^{(1)}_4) = (2 \; 2 \; 2 \; 2)
\label{eq:ex_LP_PCW_1}
\end{equation}
and 
\begin{equation}
(h^{(2)}_1, h^{(2)}_2, h^{(2)}_3, h^{(2)}_4) = (2 \; 2 \; 0 \; 0)
\label{eq:ex_LP_PCW_2}
\end{equation}
together with  
\begin{equation}
z_{1,\bldbb} = \left\{ \begin{array}{cc}
2 & \textrm{ if } \bldb = ( 2 \; 1 \; 1 \; 0 ) \\
2 & \textrm{ if } \bldb = ( 1 \; 2 \; 0 \; 1 ) \\
0 & \textrm{ otherwise, }\end{array}\right. 
\label{eq:ex_LP_PCW_3}
\end{equation}
and
\begin{equation}
z_{2,\bldbb} = \left\{ \begin{array}{cc}
2 & \textrm{ if } \bldb = ( 2 \; 0 \; 1 ) \\
2 & \textrm{ if } \bldb = ( 1 \; 1 \; 0 ) \\
0 & \textrm{ otherwise, }\end{array}\right.
\label{eq:ex_LP_PCW_4}
\end{equation}
satisfy~(\ref{eq:LP_PCW_1}) and~(\ref{eq:LP_PCW_2}), where $M=4$ in~(\ref{eq:LP_PCW_2}). 
We also obtain from~(\ref{eq:h_sum_equals_M})
\begin{equation*}
(h^{(0)}_1, h^{(0)}_2, h^{(0)}_3, h^{(0)}_4) = (0 \; 0 \; 2 \; 2) \; . 
\end{equation*}
Therefore~(\ref{eq:ex_LP_PCW_1})-(\ref{eq:ex_LP_PCW_4}) define an LP pseudocodeword, with pseudocodeword matrix  
\begin{equation}
\mathsf{H} = \left( \begin{array}{ccc}
0 & 2 & 2 \\
0 & 2 & 2 \\
2 & 2 & 0 \\
2 & 2 & 0 
\end{array} \right) \; .
\label{eq:ex_LP_PCW_matrix}
\end{equation}
The corresponding normalized LP pseudocodeword matrix is then given by 
\begin{equation}
\frac{1}{4} \cdot \mathsf{H} = \left( \begin{array}{ccc}
0 & \half & \half \\
0 & \half & \half \\
\half & \half & 0 \\
\half & \half & 0 
\end{array} \right) \; .
\label{eq:ex_LP_NPCW_matrix}
\end{equation}
Here the probabilistic interpretation of this normalized LP pseudocodeword matrix corresponds to an equiprobable distribution of symbols from $\{ 1,2 \}$ for the first two symbols in the codeword, and an equiprobable distribution of symbols from $\{ 0,1 \}$ for the last two symbols in the codeword.
\end{example}

\begin{theorem}
Assume that the all-zero codeword was transmitted.
\begin{enumerate}
\item
If the LP decoder makes a codeword error, then there exists some LP
pseudocodeword $(\bldh, \bldz)$, $\bldh \neq \zeros$, such that $\boldsymbol{\Lambda} (\bldy) \bldh ^T \le 0$.
\item
If there exists some LP pseudocodeword $(\bldh, \bldz)$, $\bldh \neq \zeros$, such that $\boldsymbol{\Lambda} (\bldy)
\bldh ^T < 0$, then the LP decoder makes a codeword error.  
\end{enumerate}
\end{theorem}

\begin{proof} 
The proof follows the lines of its counterpart in~\cite{Feldman}. 
\begin{enumerate}
\item
Let $(\bldf, \bldw)$ be the point in $\cQ$ which minimizes $\boldsymbol{\Lambda} (\bldy) \bldf^T$. Suppose there is a codeword error; then $\bldf \neq \bldzero$, and we must have $\boldsymbol{\Lambda} (\bldy) \bldf^T \le 0$.

Next, we construct the LP pseudocodeword $(\bldh, \bldz)$ as follows. Since the LP has rational coefficients, all elements of the vectors $\bldf$ and $\bldw$ must be rational. Let $M$ denote their lowest common denominator; since $\bldf \neq \bldzero$ we may have $M>0$. Now set $h_i^{(\alpha)} = M \cdot f_i^{(\alpha)}$ for all $i \in \cI$, $\alpha\in\rrrm$ and set $z_{j,\bldbb} = M  \cdot w_{j,\bldbb}$ for all $j \in \cJ$ and $\bldb \in \code_j$.

By~(\ref{eq:equation-polytope-3})-(\ref{eq:equation-polytope-5}), $(\bldh, \bldz)$ is an LP pseudocodeword and $\bldh \neq \bldzero$ since $\bldf \neq \bldzero$. Also $\boldsymbol{\Lambda}(\bldy) \bldf^T \le 0$ implies $\boldsymbol{\Lambda} (\bldy) \bldh^T \le 0$.
\item
Now, suppose that an LP pseudocodeword $(\bldh, \bldz)$ with $\bldh \neq \bldzero$ satisfies 
$\boldsymbol{\Lambda} (\bldy) \bldh^T < 0$. Since $\bldh \neq \bldzero$ we have $M>0$ in~(\ref{eq:LP_PCW_2}). Now, set $f_i^{(\alpha)} = h_i^{(\alpha)} / M $ for all $i \in \cI$, $\alpha\in\rrrm$, and set $w_{j,\bldbb} = z_{j,\bldbb} / M$ for all $j \in \cJ$ and $\bldb \in \code_j$. It is straightforward to check that $(\bldf, \bldw)$ satisfies all the constraints of the polytope 
$\cQ$. Also, $\bldh \neq \zeros$ implies $\bldf \neq \zeros$. 
Finally, $\boldsymbol{\Lambda} (\bldy) \bldh^T < 0$ implies $\boldsymbol{\Lambda} (\bldy) \bldf^T < 0$. Therefore, the LP decoder will make a codeword error. 
\end{enumerate}
\end{proof}

\section{Equivalence Between Pseudocodeword Concepts}
\label{sec:eq-pseudocodewords}
\subsection{Tanner Graphs and Graph-Cover Pseudocodewords}
The Tanner graph of a linear code $\code$ over $\rrr$ is an equivalent characterization of the code's parity-check matrix $\cH$. The Tanner graph $\graph = (\cV, \cE)$ has vertex set $\cV = \{u_1, u_2, \cdots, u_n \} \cup  
\{v_1, v_2, \cdots, v_m \}$, and there is an edge between $u_i$ and $v_j$ if and only if $\cH_{j,i} \neq 0$. This edge is labelled with the value $\cH_{j,i}$. 
We denote by $\cN(v)$ the set of neighbors of a vertex $v\in\cV$.

For any word $\bldc  = (c_1, c_2, \cdots, c_n) \in \rrr^n$, the Tanner graph allows an equivalent graphical statement of the condition $c\in\code_j$ for each $j\in\cJ$, as follows. 
The variable vertex $u_i$ is labelled with the value $c_i$ for each $i \in \cI$. Equation~(\ref{eq:parity_check_satisfied}) (or~(\ref{eq:word_in_SPC_code})) is then equivalent to the condition that for vertex $v_j$, the sum, over all vertices in $\cN(v_j)$, of the vertex labels multiplied by the corresponding edge labels is zero. This graphical means of checking whether a parity-check is satisfied by $\bldc \in\rrr^n$ will be useful when defining graph-cover pseudocodewords later in this section.

To illustrate this concept, Figure~\ref{cap:Tanner_graph} shows the Tanner graph for the codeword 
$\bldc = ( 1 \; 0 \; 2 \; 1 )$ of the example $[4,2]$ code over $\mathbb{Z}_3$ defined by the parity-check matrix~(\ref{eq:example_PCM}). In Figure \ref{cap:Tanner_graph}, edge labels are shown in square brackets, and vertex labels in round brackets. The reader may check that for each parity-check $j=1,2$, the sum, over all vertices in $\cN(v_j)$, of the vertex labels multiplied by the corresponding edge labels is zero.

\begin{figure}
\begin{center}\includegraphics[%
  width=0.5\columnwidth, keepaspectratio]{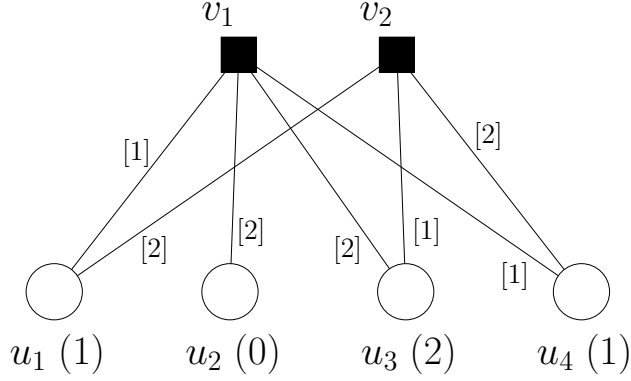}\end{center}
\caption{Tanner graph for the example $[4,2]$ code over $\mathbb{Z}_3$. Edge labels are shown in square brackets, and vertex labels in round brackets. For each parity-check $j$, the sum, over all vertices in $\cN(v_j)$, of the vertex labels multiplied by the corresponding edge labels is zero; therefore all parity-checks are satisfied.\label{cap:Tanner_graph}}
\end{figure}

We next define what is meant by a finite cover of a Tanner graph.

\begin{definition}
(\cite{KV-characterization})
A graph $\tilde{\graph} = (\tilde{\cV}, \tilde{\cE})$ is a \emph{finite cover} of the Tanner graph $\graph = (\cV, \cE)$ 
if there exists a mapping $\Pi: \tilde{\cV} \longrightarrow \cV$ which is a graph homomorphism
($\Pi$ takes adjacent vertices of $\tilde{\graph}$ to adjacent vertices of $\graph$), such that 
for every vertex $v \in \graph$ and every $\tilde{v} \in \Pi^{-1}(v)$, the neighborhood $\cN(\tilde{v})$ 
of $\tilde{v}$ (including edge labels) is mapped bijectively to $\cN(v)$. 
\end{definition}

\begin{definition}
(\cite{KV-characterization})
A cover of the graph $\graph$ is said to have degree $M$, where $M$ is a positive integer, if $|\Pi^{-1}(v)| = M$
for every vertex $v \in \cV$. We refer to such a cover graph as an $\mit{M}$\emph{-cover} of $\graph$.  
\end{definition}

Fix some positive integer $M$. Let $\tilde{\graph} = (\tilde{\cV}, \tilde{\cE})$ be an $M$-cover   
of the Tanner graph $\graph = (\cV, \cE)$ representing 
the code $\code$ with parity-check matrix $\cH$. 
The vertices in the set $\Pi^{-1} (u_i)$ are called \emph{copies} of $u_i$ and are denoted $\{ u_{i,1}, u_{i,2}, \cdots, u_{i,M} \}$, where $i\in\cI$. Similarly, the vertices in the set $\Pi^{-1} (v_j)$ are called \emph{copies} of $v_j$ and are denoted $\{ v_{j,1}, v_{j,2}, \cdots, v_{j,M} \}$, where $j\in\cJ$. 

Less formally, given a code $\code$ with parity-check matrix $\cH$ and corresponding Tanner graph $\graph$, an $M$-cover of $\graph$ is a graph whose vertex set consists of $M$ copies of $u_i$ and $M$ copies of $v_j$, such that for each $j\in\cJ$, $i\in\cI_j$, the $M$ copies of $u_i$ and the $M$ copies of $v_j$ are connected in an arbitrary one-to-one fashion, with edges labelled by the value $H_{j,i}$.

For any $M\ge 1$, a \emph{graph-cover pseudocodeword} is a labelling of vertices of the $M$-cover graph with values from $\rrr$ such that all parity-checks are satisfied. We denote the label of $u_{i,l}$ by $p_{i,l}$ for each $i\in\cI$, $\ell = 1, 2, \cdots, M$, and we may then write the graph-cover pseudocodeword in vector form as 
\begin{equation*}\bldp = ( p_{1,1}, p_{1,2}, \cdots, p_{1,M}, p_{2,1}, p_{2,2}, \cdots, p_{2,M}, \cdots, p_{n,1}, p_{n,2}, \cdots, p_{n, M} ) \; .
\end{equation*}
It is easily seen that $\bldp$ belongs to a linear code $\tilde{\code}$ of length $Mn$ over $\rrr$, 
defined by an $Mm \times Mn$ parity-check matrix $\tilde{\cH}$. To construct $\tilde{\cH}$, for
$1 \le i^*,j^* \le M$ and $i \in \cI$, $j \in \cJ$, we let  
$i' = (i-1) M + i^*, j' = (j-1) M + j^*$, and so
\[
\tilde{\cH}_{j',i'} = \left\{ \begin{array}{cl}
\cH_{j,i} & \mbox{if } u_{i,i^*} \in \cN(v_{j,j^*}) \\
0 & \mbox{otherwise} 
\end{array} \right. \; .
\]
It may be seen that $\tilde{\graph}$ is the Tanner graph of the code $\tilde{\code}$ corresponding to the parity-check matrix $\tilde{\cH}$.

We also define the $n \times q$ \emph{graph-cover pseudocodeword matrix} 
\[
\cP = \Big( m_i^{(\alpha)} \Big)_{i \in \cI; \, \alpha\in\rrr} \; ,  
\] 
where 
\[
m_i^{(\alpha)} = \left| \{ \ell \in \{ 1, 2, \cdots, M \} \; : \; p_{i,\ell} = \alpha \} \right| \ge 0 \; , 
\]
for $i\in\cI$, $\alpha\in\rrr$, i.e. $m_i^{(\alpha)}$ is equal to the number of copies of $u_i$ which are labelled with $\alpha$, for each $i\in\cI$, $\alpha\in\rrr$. The \emph{normalized graph-cover pseudocodeword matrix}
is defined as $(1/M) \cdot \cP$. This matrix representation is similar to that defined in~\cite{Kelley-Sridhara-ISIT-2006}. 
Note that the $i$-th row of the normalized graph-cover pseudocodeword matrix (for $i \in \cI$) can be viewed as a probability distribution for the $i$-th coded symbol $c_i \in \rrr$, in a similar manner to the case of the normalized LP pseudocodeword matrix.

Another representation, which we shall use in Section \ref{sec:cascaded_polytope}, is the \emph{graph-cover pseudocodeword vector} $\bldm = (\bldm_i)_{i \in \cI} $ where $\bldm_i = (m_i^{(\alpha)})_{\alpha\in\rrrm}$ for each $i \in \cI$. Correspondingly, the \emph{normalized graph-cover pseudocodeword vector} is given by $(1/M) \cdot \bldm \in \mathbb{R}^{(q-1)n}$.

It is easily seen that for any $\bldc\in\code$, the labelling of $u_{i,l}$ by the value $c_i$ for all $i\in\cI$, $\ell = 1, 2, \cdots, M$ trivially yields a pseudocodeword for all $M$-covers of $\graph$, $M\ge 1$. However, non-trivial pseudocodewords exist in general. 
 
\begin{example}
To illustrate these concepts, a graph-cover pseudocodeword in shown in Figure~\ref{cap:cover_graph} for the example $[4,2]$ code over $\mathbb{Z}_3$ defined by the parity-check matrix~(\ref{eq:example_PCM}). Here the degree of the cover graph is $M=4$, and we have 
\[
\bldp = ( 1 \; 1 \; 2 \; 2 \; | \;  1 \; 1 \; 2 \; 2 \; | \; 0 \; 0 \; 1 \; 1 \; | \; 0 \; 0 \; 1 \; 1) \; ,
\]
and the parity-check matrix of the code $\tilde{\code}$ is given by  
\begin{eqnarray*}
\tilde{\cH} = 
\left( \begin{array}{c|c|c|c}
0 \; 0 \; 1 \; 0 & 2 \; 0 \; 0 \; 0 & 0 \; 0 \; 2 \; 0 & 1 \; 0 \; 0 \; 0  \\
0 \; 0 \; 0 \; 1 & 0 \; 2 \; 0 \; 0 & 0 \; 0 \; 0 \; 2 & 0 \; 1 \; 0 \; 0  \\
1 \; 0 \; 0 \; 0 & 0 \; 0 \; 2 \; 0 & 2 \; 0 \; 0 \; 0 & 0 \; 0 \; 1 \; 0  \\
0 \; 1 \; 0 \; 0 & 0 \; 0 \; 0 \; 2 & 0 \; 2 \; 0 \; 0 & 0 \; 0 \; 0 \; 1  \\
\hline
0 \; 0 \; 2 \; 0 & 0 \; 0 \; 0 \; 0 & 1 \; 0 \; 0 \; 0 & 0 \; 0 \; 2 \; 0  \\
0 \; 0 \; 0 \; 2 & 0 \; 0 \; 0 \; 0 & 0 \; 1 \; 0 \; 0 & 0 \; 0 \; 0 \; 2  \\
2 \; 0 \; 0 \; 0 & 0 \; 0 \; 0 \; 0 & 0 \; 0 \; 1 \; 0 & 2 \; 0 \; 0 \; 0  \\
0 \; 2 \; 0 \; 0 & 0 \; 0 \; 0 \; 0 & 0 \; 0 \; 0 \; 1 & 0 \; 2 \; 0 \; 0  
\end{array} \right)
\end{eqnarray*}
Also, the graph-cover pseudocodeword matrix corresponding to $\bldp$ is 
\begin{equation}
\cP = \left( \begin{array}{ccc}
0 & 2 & 2 \\
0 & 2 & 2 \\
2 & 2 & 0 \\
2 & 2 & 0 
\end{array} \right) \; ,
\label{eq:ex_GC_PCW_matrix}
\end{equation}
and the normalized graph-cover pseudocodeword matrix is 
\[
\frac{1}{4} \cdot \cP \; . 
\] 
The graph-cover pseudocodeword vector corresponding to $\bldp$ is 
\begin{equation*}
\bldm = \left( \; 2 \; 2 \; | \; 2 \; 2 \; | \; 2 \; 0 \; | \; 2 \; 0 \; \right) \; ,  
\end{equation*}
and the normalized graph-cover pseudocodeword vector is 
\[
\frac{1}{4} \cdot \bldm \; . 
\] 
\begin{figure*}
\begin{center}\includegraphics[%
  width=1.0\columnwidth,
  keepaspectratio]{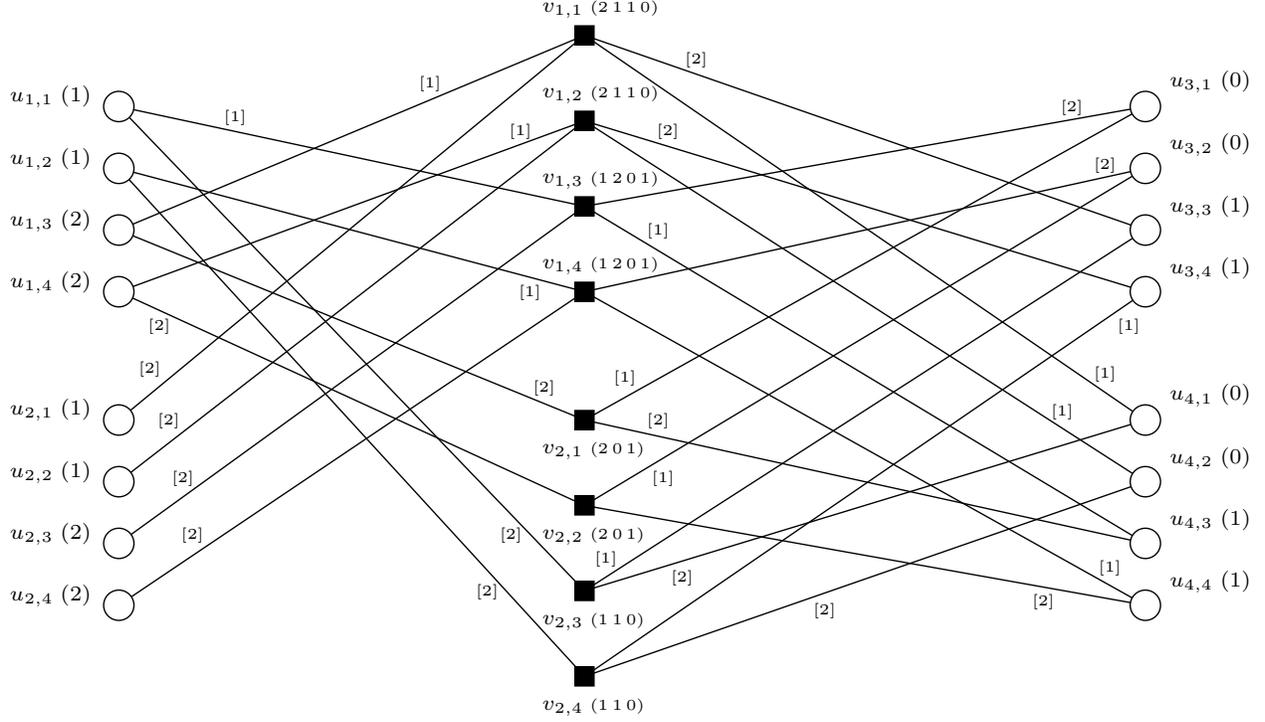}
  \end{center}
\caption{Cover graph of degree $4$ and corresponding graph-cover pseudocodeword for the example $[4,2]$ code over $\mathbb{Z}_3$ with parity-check matrix given by~(\ref{eq:example_PCM}). Edge labels are shown in square brackets, and vertex labels in round brackets. This graph-cover pseudocodeword corresponds to the LP pseudocodeword described by~(\ref{eq:ex_LP_PCW_1})-(\ref{eq:ex_LP_PCW_4}) via the correspondence described in the proof of Theorem \ref{thm:PCW_equivalence}.  \label{cap:cover_graph}}
\end{figure*}
\end{example}

\subsection{Equivalence between LP Pseudocodewords and Graph-Cover Pseudocodewords}
In this section, we show the equivalence between the set of LP pseudocodewords and the set of
graph-cover pseudocodewords. The result is summarized in the following theorem. 

\begin{theorem}
\label{thm:PCW_equivalence}
Let $\code$ be a linear code over the ring $\rrr$ with parity-check matrix $\cH$ and corresponding Tanner graph $\cG$. Then, there exists an LP pseudocodeword $(\bldh, \bldz)$ with pseudocodeword matrix $\mathsf{H}$ if and only if there 
exists a graph-cover pseudocodeword for some $M$-cover of $\cG$ with the same pseudocodeword matrix. 
\end{theorem}

\begin{proof} 
\begin{enumerate}
\item
Let $(\bldh, \bldz)$ be an LP pseudocodeword of $\code$, and   
let $\graph = (\cV, \cE)$ be the Tanner graph associated with the parity-check matrix $\cH$. 
We construct an $M$-cover $\tilde{\graph} = (\tilde{\cV}, \tilde{\cE})$, where $M = \sum_{\bldbb \in \code_j} z_{j,\bldbb}$, and corresponding graph-cover pseudocodeword, as follows. We begin with the vertex set, which consists of $M$ copies of $u_i$, $i\in\cI$, and $M$ copies of $v_j$, $j\in\cJ$. Then we proceed as follows:

\begin{itemize}
\item
Label $h_i^{(\alpha)}$ copies of $u_i$ with the value $\alpha$, for each $i \in \cI$, $\alpha \in \rrr$. By~(\ref{eq:h_sum_equals_M}), all copies of $u_i$ are labelled. 

\item
Label $z_{j,\bldbb}$ copies of $v_j$ with the value $\bldb$, for every $j \in \cJ$, $\bldb \in \code_j$. By~(\ref{eq:LP_PCW_2}), all copies of $v_j$ are labelled.

\item
Next, let $T_i^{(\alpha)}$ denote the set of copies of $u_i$ labelled with the value $\alpha$, for $i\in\cI$, $\alpha\in\rrr$. Also, for all $i\in\cI$, $j\in\cJ$, $\alpha\in\rrr$, let $R_{i,j}^{(\alpha)}$ denote the set of copies of $v_j$ whose label satisfies $b_i=\alpha$. The vertices in $T_i^{(\alpha)}$ and the vertices in $R_{i,j}^{(\alpha)}$ are then connected by edges in an arbitrary one-to-one fashion, for every $j\in\cJ$, $i\in\cI_j$, $\alpha\in\rrr$. All of these edges are labelled with the value $H_{j,i}$.

First, we note that this is possible because
\begin{eqnarray*}
| T_i^{(\alpha)} | & = & h_i^{(\alpha)} \\
& = & \sum_{\bldbb \in \code_j, \; b_i=\alpha} z_{j,\bldbb} \\
& = & | R_{i,j}^{(\alpha)} | 
\end{eqnarray*}
for every $j\in\cJ$, $i\in\cI_j$, $\alpha\in\rrr$. Here we have used~(\ref{eq:LP_PCW_1})). 

Second, we note that all checks are satisfied by this labelling. For $j\in\cJ$, consider any copy of $v_j$ with label $\bldb$. By construction of the graph, the sum, over all vertices in $\cN(v_j)$, of the vertex labels multiplied by the corresponding edge labels is 
\[
\sum_{i\in\cI_j} b_i \cdot \cH_{j,i} \; ,
\]
which is zero because $\bldb \in \code_j$.
Therefore, this vertex labelling yields a graph-cover pseudocodeword of the code $\code$ with parity-check matrix $\cH$.  
\end{itemize}
\item
Now suppose that there exists a graph-cover pseudocodeword corresponding to some 
$M$-cover of the Tanner graph $\cG$ of $\code$. Then, 
\begin{itemize}
\item
{\bf Step 1:} for every $i \in \cI$, and for every $\alpha \in \rrrm$, 
we define $h_i^{(\alpha)}$ to be the number of copies of $u_i$ labelled with the value $\alpha$.
\item
{\bf Step 2:} for every copy of $v_j$, $j\in\cJ$, label the copy with the word $\bldb$, where $b_i$ is equal to the label on the neighbouring copy of $u_i$, $i\in\cI_j$. Then, for every $j\in\cJ$, $\bldb\in\code_j$, we define $z_{j,\bldbb}$ to be the number of copies of $v_j$ labelled with the word $\bldb$. 
\end{itemize}

Step 2 ensures that $z_{j,\bldbb}$ are nonnegative integers for all $j \in \cJ$ and $\bldb \in \code_j$, and that~(\ref{eq:LP_PCW_2}) holds. Also, to show that~(\ref{eq:LP_PCW_1}) holds, we reason as follows. The right-hand side of~(\ref{eq:LP_PCW_1}) counts the number of copies of $v_j$ whose labels $\bldb$ satisfy $b_i=\alpha$. By step 2, this is equal to the number of copies of $u_i$ labelled with $\alpha$, which by step 1 is equal to the left-hand side of~(\ref{eq:LP_PCW_1}). Therefore, $(\bldh, \bldz)$ is an LP pseudocodeword of the code $\code$ with parity-check matrix $\cH$.  

\end{enumerate}
\end{proof}

As an illustration of the correspondences described in this proof, consider the example $[4,2]$ code over $\mathbb{Z}_3$ defined by the parity-check matrix~(\ref{eq:example_PCM}). First, note that the LP pseudocodeword of~(\ref{eq:ex_LP_PCW_1})-(\ref{eq:ex_LP_PCW_4}) and the graph-cover pseudocodeword of Figure~\ref{cap:cover_graph} have the same pseudocodeword matrix, via~(\ref{eq:ex_LP_PCW_matrix}) and~(\ref{eq:ex_GC_PCW_matrix}). Indeed, the reader may check that each pseudocodeword may be derived from the other using the correspondences described in the proof of Theorem \ref{thm:PCW_equivalence}.

The next corollary follows immediately from Theorem~\ref{thm:PCW_equivalence}. 
\begin{corollary}
\label{thm:PCW_vector_equivalence}
Let $\code$ be a linear code over the ring $\rrr$ with parity-check matrix $\cH$ and corresponding Tanner graph $\cG$. Then, there exists a (normalized) LP pseudocodeword $(\bldh, \bldz)$ if and only if there 
exists a graph-cover pseudocodeword for some $M$-cover of $\cG$ with (normalized) graph-cover pseudocodeword vector $\bldh$. 
\end{corollary}
Note that this corollary contains two different equivalences, one for normalized objects and the other for non-normalized ones. 

\section{Alternative Polytope Representation}
\label{sec:alternative_polytope}

In this section, we present an alternative polytope for use with linear-programming decoding. This polytope may be regarded as a generalization of the ``high-density polytope'' defined in \cite{Feldman}. As we show in this section, the new polytope may under some circumstances yield a complexity advantage over the polytope 
of Section~\ref{sec:lp}. In the sequel, we will analyze the properties of this polytope.  

First, we introduce some convenient notation and definitions. Recall that the ring $\rrr$ contains $q-1$ non-zero elements; correspondingly, for vectors $\bldk\in\nn^{q-1}$, we adopt the notation
\[
\bldk = (k_{\alpha})_{\alpha\in\rrrm}
\]
Now, for any $j\in\cJ$, we define the mapping
\begin{eqnarray*}
\bldkappa_j \; : \; \code_j & \longrightarrow & \nn^{q-1} \; , \\
\bldb & \mapsto & \bldkappa_j(\bldb)
\end{eqnarray*}
defined by 
\[
( \bldkappa_j (\bldb) )_\alpha = \left| \left\{ i \in \cI_j \; : \; b_i \cdot \cH_{j,i} = \alpha \right\} 
\right|
\]
for all $\alpha \in \rrrm$. We may then characterize the image of $\bldkappa_j$, which we denote by $\cT_j$, as 
\[
\cT_j = \left\{ \bldk \in \nn^{q-1} \; : \; \sum_{\alpha \in \rrrm} \alpha \cdot k_\alpha = 0 \mbox{ and } 
\sum_{\alpha \in \rrrm} k_\alpha \le d_j \right\} \; ,   
\]
for each $j\in\cJ$, where, for any $k\in\nn$, $\alpha\in\rrr$, 
\[
\alpha \cdot k = \left\{ \begin{array}{cc}
0 & \textrm{ if } k = 0 \\
\alpha + \cdots + \alpha & \textrm{ if } k > 0 \textrm{ (} k \textrm{ terms in sum)}\end{array}\right. 
\]
Note that $\bldkappa_j$ is not a bijection, in general. We say that a local codeword $\bldb \in \code_j$ is $\bldk$-constrained over $\code_j$ if $\bldkappa_j (\bldb) = \bldk$. 

Next, for any index set $\Gamma\subseteq\cI$, we introduce the following definitions. Let $N = \left| \Gamma \right|$. We define the single-parity-check-code, over vectors indexed by $\Gamma$, by
\begin{equation}
\code_\Gamma = \left\{ \blda = (a_i)_{i\in\Gamma} \in \rrr^N \; : \; \sum_{i\in\Gamma} a_i = 0 \right\} \; . 
\end{equation}
Also define a mapping $\bldkappa_\Gamma \; : \; \code_\Gamma \longrightarrow \nn^{q-1}$ by
\[
\left(\bldkappa_\Gamma(\blda)\right)_\alpha = \left| \left\{ i\in\Gamma \; : \; a_i = \alpha \right\} \right| \; ,  
\] and define, for $\bldk\in\cT_j$, 
\[
\code_\Gamma^{(\bldkk)} = \{ \blda \in \code_\Gamma \; : \; \bldkappa_\Gamma(\blda) = \bldk \} \; .
\]

Below, we define a new polytope for decoding. Recall that 
$\bldy = (y_1, y_2, \cdots, y_n) \in \Sigma^n$ stands for the received (corrupted) word. 
In the sequel, we make use of the following variables:
\begin{itemize}
\item
For all $i \in \cI$ and all $\alpha \in \rrrm$, we have a variable $f_i^{(\alpha)}$. This variable is an indicator of the event $y_i= \alpha$. 
\item
For all $j \in \cJ$ and $\bldk \in \cT_j$, we have a variable $\sigma_{j,\bldkk}$. 
Similarly to its counterpart in~\cite{Feldman}, this variable indicates the contribution to parity-check $j$ of 
$\bldk$-constrained local codewords over $\code_j$.  
\item
For all $j \in \cJ$, $i\in\cI_j$, $\bldk \in \cT_j$, $\alpha \in \rrrm$, we have 
a variable $z^{(\alpha)}_{i,j,\bldkk}$. 
This variable indicates the portion of $f_i^{(\alpha)}$ assigned to $\bldk$-constrained local codewords over $\code_j$.
\end{itemize}

Motivated by these variable definitions, for all $j \in \cJ$ we impose the following set of constraints: 
\begin{equation}
\forall i \in \cI_j, \forall \alpha \in \rrrm,  
\qquad f_i^{(\alpha)} = \sum_{\bldkk \in \cT_j} z_{i,j,\bldkk}^{(\alpha)} \; . 
\label{eq:LP_1}
\end{equation}
\begin{equation}
\sum_{\bldkk \in \cT_j} \sigma_{j,\bldkk} = 1 \; .
\label{eq:LP_2}
\end{equation}
\begin{equation}
\forall \bldk \in \cT_j, \forall \alpha \in \rrrm,   
\qquad \sum_{i \in \cI_j, \; \beta\in\rrrm, \; \beta\cH_{j,i}=\alpha} z_{i,j,\bldkk}^{(\beta)} = k_\alpha \cdot \sigma_{j,\bldkk} \; . 
\label{eq:LP_3}
\end{equation}
\begin{equation}
\forall i \in \cI_j, \forall \bldk \in \cT_j, \forall \alpha \in \rrrm,
\qquad z^{(\alpha)}_{i,j,\bldkk} \ge 0 \; . 
\label{eq:LP_4}
\end{equation}
\begin{equation}
\forall i \in \cI_j, \forall \bldk \in \cT_j,
\qquad \sum_{\alpha \in \rrrm} \; \sum_{\beta\in\rrrm, \; \beta\cH_{j,i}=\alpha} z^{(\beta)}_{i,j,\bldkk} \le \sigma_{j,\bldkk} \; . 
\label{eq:LP_5}
\end{equation}
We note that the further constraints
\begin{equation}
\forall i \in \cI, \forall \alpha \in \rrrm,  
\qquad 0 \le f_i^{(\alpha)} \le 1 \; ,
\label{eq:LP_6}
\end{equation}
\begin{equation}
\forall j \in \cJ, \forall \bldk \in \cT_j,  
\qquad 0 \le \sigma_{j,\bldkk} \le 1 \; , 
\label{eq:LP_7} 
\end{equation}
and
\begin{equation}
\forall j \in \cJ, \forall i \in \cI_j, \forall \bldk \in \cT_j, \forall \alpha \in \rrrm,
\qquad z_{i,j,\bldkk}^{(\alpha)} \le \sigma_{j,\bldkk} \; ,
\label{eq:LP_8} 
\end{equation}
follow from constraints~(\ref{eq:LP_1})-(\ref{eq:LP_5}).
We denote by $\cU$ the polytope formed by
constraints~(\ref{eq:LP_1})-(\ref{eq:LP_5}).  

Let $T = \max_{j \in \cJ} |\cT_j|$.  
Then, upper bounds on the number of variables and constraints in this LP are given by $n(q-1) + m(d(q-1)+1) T$ and 
$m(d(q-1)+1) + m((d+1)(q-1)+d) T$, respectively. Since $T \le {d+q-1 \choose d}$, the number of variables
and constraints are $O(m q \cdot d^{q})$, which, for many families of codes, 
is significantly lower than the corresponding complexity for polytope $\cQ$. 

For notational simplicity in proofs in this section, it is convenient to define a new set of variables as follows:
\begin{equation}
\forall j \in \cJ, \forall i \in \cI_j, \forall \bldk \in \cT_j, \forall \alpha \in \rrrm, 
\qquad \tau^{(\alpha)}_{i,j,\bldkk} = \sum_{\beta\in\rrrm, \; \beta\cH_{j,i}=\alpha} z^{(\beta)}_{i,j,\bldkk} \; . 
\label{eq:tau_definition}
\end{equation}
Then constraints~(\ref{eq:LP_3}) and~(\ref{eq:LP_5}) may be rewritten as
\begin{equation}
\forall j \in \cJ, \bldk \in \cT_j, \forall \alpha \in \rrrm,    
\qquad \sum_{i \in \cI_j} \tau_{i,j,\bldkk}^{(\alpha)} = k_\alpha \cdot \sigma_{j,\bldkk} \; . 
\label{eq:LP_3a}
\end{equation}
and
\begin{equation}
\forall j \in \cJ, \forall i \in \cI_j, \forall \bldk \in \cT_j,
\qquad 0 \le \sum_{\alpha \in \rrrm} \; \tau^{(\alpha)}_{i,j,\bldkk} \le \sigma_{j,\bldkk} \; . 
\label{eq:LP_5a}
\end{equation}
Note that the variables $\bldtau$ do not form part of the LP description, and therefore do not contribute to its complexity. However these variables will provide a convenient notational shorthand for proving results in this section.

We will prove that optimizing the cost function~(\ref{eq:object-function}) over this new polytope is equivalent to optimizing over $\cQ$. First, we state the following proposition, which will be necessary to prove this result.

\begin{proposition}
Let $M\in\nn$ and $\bldk\in\nn^{q-1}$. Also let $\Gamma\subseteq\cI$. Assume that for each $\alpha\in\rrrm$, we have a set of nonnegative integers $\cX^{(\alpha)} = \{ x^{(\alpha)}_i \; : \; i\in\Gamma \}$ and that together these satisfy the constraints 
\begin{equation}
\sum_{i\in\Gamma} x^{(\alpha)}_i = k_\alpha M 
\label{eq:lemma-14-req-1}
\end{equation}
for all $\alpha \in \rrrm$, 
and 
\begin{equation}
\sum_{\alpha \in \rrrm} x^{(\alpha)}_i \le M
\label{eq:lemma-14-req-2}
\end{equation} 
for all $i\in\Gamma$. 

Then, there exist nonnegative integers 
$\left\{ w_\bldaa \; : \; \blda \in \code_\Gamma^{(\bldkk)} \right\}$
such that
\begin{enumerate}
\item
\begin{equation}
\sum_{\bldaa \in \code_\Gamma^{(\bldkk)}} w_\bldaa = M \; .  
\label{eq:lemma14-claim-1}
\end{equation}
\item
For all $\alpha \in \rrrm$, $i\in\Gamma$, 
\begin{equation}
x^{(\alpha)}_i = 
\sum_{\bldaa \in \code_\Gamma^{(\bldkk)}, \; a_i = \alpha} w_\bldaa \; . 
\label{eq:lemma14-claim-2}
\end{equation}
\end{enumerate}
\label{prop:lemma-14}
\end{proposition}

The proof of this proposition appears in the Appendix. 
We now prove the main result.

\begin{theorem}
The set $\bar{\cU} = \{ \bldf : \exists \; \bldsigma, \bldz \mbox{ s.t. } (\bldf, \bldsigma, \bldz) \in \cU \}$
is equal to the set $\bar{\cQ} = 
\{ \bldf : \exists \; \bldw \mbox{ s.t. } (\bldf, \bldw) \in \cQ \}$. 
Therefore, optimizing the linear cost function~(\ref{eq:object-function}) over $\cU$ is equivalent to optimizing over $\cQ$. 
\end{theorem}

\begin{proof}
\begin{enumerate}
\item
Suppose, $(\bldf, \bldw) \in \cQ$. For all $j \in \cJ, \bldk \in \cT_j$, we define
\[
\sigma_{j,\bldkk} = \sum_{\bldbb \in \code_j, \; \bldsubkappa_j(\bldbb) = \bldkk} w_{j,\bldbb} \; ,  
\]
and for all $j \in \cJ, \; i \in \cI_j, \; \bldk \in \cT_j$, $\alpha \in \rrrm$, we define
\[
z^{(\alpha)}_{i,j,\bldkk} = 
\sum_{\bldbb \in \code_j, \; \bldsubkappa_j(\bldbb) = \bldkk, \; b_i = \alpha} w_{j,\bldbb} \; ,  
\]

It is straightforward to check that constraints~(\ref{eq:LP_4}) and~(\ref{eq:LP_5}) are satisfied by these definitions. 

For every $j \in \cJ, \; i \in \cI_j, \; \alpha \in \rrrm$, 
we have by (\ref{eq:equation-polytope-5})
\begin{eqnarray*}
f^{(\alpha)}_i & = & \sum_{\bldbb \in \code_j, \; b_i = \alpha} w_{j,\bldbb} \\
    & = & \sum_{\bldkk \in \cT_j} \quad \sum_{\bldbb \in \code_j, \; \bldsubkappa_j(\bldbb) = \bldkk, \; b_i = \alpha} w_{j,\bldbb} \\
    & = & \sum_{\bldkk \in \cT_j} z^{(\alpha)}_{i,j,\bldkk} \; , 
\end{eqnarray*}
and thus constraint~(\ref{eq:LP_1}) is satisfied. 

Next, for every $j \in \cJ$, 
we have by~(\ref{eq:equation-polytope-4})
\begin{eqnarray*}
1 & = & \sum_{\bldbb \in \code_j} w_{j,\bldbb} \\
  & = & \sum_{\bldkk \in \cT_j} \quad \sum_{\bldbb \in \code_j, \bldsubkappa_j(\bldbb) = \bldkk} w_{j,\bldbb} \\
  & = & \sum_{\bldkk \in \cT_j} \sigma_{j, \bldkk} \; , 
\end{eqnarray*}
and thus constraint~(\ref{eq:LP_2}) is satisfied. 

Finally, for every $j \in \cJ, \; \bldk \in \cT_j, \; \alpha \in \rrrm$, 
\begin{eqnarray*}
&& \hspace{-7ex} \sum_{i \in \cI_j, \; \beta\in\rrrm, \; \beta\cH_{j,i}=\alpha} z_{i,j,\bldkk}^{(\beta)} \\
& = & \sum_{i \in \cI_j, \; \beta\in\rrrm, \; \beta\cH_{j,i}=\alpha} \quad \sum_{\bldbb \in \code_j, \; \bldsubkappa_j(\bldbb) = \bldkk, \; b_i = \beta} w_{j,\bldbb} \\
  & = & \sum_{\bldbb \in \code_j, \; \bldsubkappa_j(\bldbb) = \bldkk} \quad 
        \sum_{i \in \cI_j, \; b_i \cH_{j,i} = \alpha} w_{j,\bldbb} \\
  & = & \sum_{\bldbb \in \code_j, \; \bldsubkappa_j(\bldbb) = \bldkk} k_\alpha \cdot w_{j, \bldbb} \\
  & = & k_\alpha \cdot \sigma_{j,\bldkk} \; .
\end{eqnarray*}
Thus, constraint~(\ref{eq:LP_3}) is also satisfied. 
This completes the proof of the first part of the theorem.
\item
Now assume $(\bldf, \bldsigma, \bldz)$ is a vertex of the polytope $\cU$, and so all variables are rational, as are the variables $\bldtau$. 
Next, fix some $j \in \cJ, \bldk \in \cT_j$, and consider the sets
\[
\cX_0^{(\alpha)}= \left\{ \frac{\tau^{(\alpha)}_{i,j,\bldkk}}{\sigma_{j,\bldkk}} 
\; : \; i\in\cI_j \right\} \; . 
\] 
for $\alpha\in\rrrm$. By constraint~(\ref{eq:LP_5a}), for each $\alpha\in\rrrm$, all the values in the set $\cX_0^{(\alpha)}$ are rational numbers between 0 and 1. Let $\mu$ be the lowest common denominator of all the numbers in all the sets $\cX_0^{(\alpha)}$, $\alpha\in\rrrm$. Let 
\[
\cX^{(\alpha)}= \left\{ \mu \cdot \frac{\tau^{(\alpha)}_{i,j,\bldkk}}{\sigma_{j,\bldkk}} 
\; : \; i\in\cI_j \right\} \; ,
\] 
for each $\alpha\in\rrrm$. The sets $\cX^{(\alpha)}$ consist of integers between 0 and $\mu$. By constraint~(\ref{eq:LP_3a}), 
we must have that for every $\alpha \in \rrrm$, the sum of the elements in $\cX^{(\alpha)}$ is equal to $k_\alpha \mu$. 
By constraint~(\ref{eq:LP_5a}), we have 
\[
\sum_{\alpha \in \rrrm} \mu \cdot \frac{\tau^{(\alpha)}_{i,j,\bldkk}}{\sigma_{j,\bldkk}} \le \mu \; 
\]
for all $i \in \cI_j$. 

We now apply the result of Proposition~\ref{prop:lemma-14} with $\Gamma = \cI_j$, $M=\mu$ and with the sets $\cX^{(\alpha)}$ defined as above (here $N=d_j$). Set the variables $\{ w_\bldaa \; : \; \blda \in \code_\Gamma^{(\bldkk)} \}$ according to Proposition~\ref{prop:lemma-14}. 

Next, for $\bldk\in\cT_j$, we show how to define the variables 
$\{ w'_\bldbb \; : \; \bldb \in \code_j, \; \bldkappa_j(\bldb) = \bldk \}$. 
Initially, we set $w'_\bldbb = 0$ for all  $\bldb \in \code_j, \; \bldkappa_j(\bldb) = \bldk$. 
Observe that the values $\mu \cdot z^{(\beta)}_{i,j,\bldkk} /\sigma_{j,\bldkk}$ are 
nonnegative integers for every $i \in \cI, \; j \in \cJ, \; \bldk \in \cT_j, \; \beta \in \rrrm$.

For every $\blda \in \code_\Gamma^{(\bldkk)}$, 
we define $w_\bldaa$ words $\bldb^{(1)},$ $\bldb^{(1)}, \cdots, \bldb^{(w_\bldaa)} \in \code_j$. 
Assume some ordering on the elements $\beta \in \rrrm$ satisfying $\beta \cH_{j,i} = a_i$,
namely $\beta_1, \beta_2, \cdots, \beta_{\ell_0}$ for some positive integer $\ell_0$.
For $i \in \cI_j$, $\bldb_{i}^{(\ell)}$ ($\ell =1 ,2, \cdots, w_\bldaa$) is defined as follows: 
$\bldb_{i}^{(\ell)}$ is equal to $\beta_1$ for the first $\mu \cdot z^{(\beta_1)}_{i,j,\bldkk} /\sigma_{j,\bldkk}$ 
words $\bldb^{(1)}, \bldb^{(2)}, \cdots, \bldb^{(w_\bldaa)}$; $\bldb_{i}^{(\ell)}$ is equal to $\beta_2$ for the next 
$\mu \cdot z^{(\beta_2)}_{i,j,\bldkk} /\sigma_{j,\bldkk}$
words, and so on. 
For every $\bldb \in \code_j$ we define 
\[
w'_\bldbb = \left| \left\{ i \in \{ 1, 2, \cdots, w_\bldaa \} \; : \; \bldb^{(i)} = \bldb \right\} \right| \; . 
\] 

Finally, for every $\bldb \in \code_j, \bldkappa_j(\bldb) = \bldk$, we define
\[
w_{j,\bldbb} = \frac{\sigma_{j, \bldkk}}{\mu} \cdot w'_\bldbb \; . 
\]
 
Using Proposition~\ref{prop:lemma-14}, 
\[
\sum_{\bldaa \in \code_\Gamma^{(\bldkk)}, \; a_i = \alpha} w_\bldaa  = 
\mu \cdot \frac{\tau^{(\alpha)}_{i,j,\bldkk}}{\sigma_{j,\bldkk}} 
= \sum_{\beta \; : \; \beta \cH_{j,i} = \alpha} \mu \cdot \frac{z^{(\beta)}_{i,j,\bldkk}}{\sigma_{j,\bldkk}} \; , 
\]
and so all $\bldb^{(1)}, \bldb^{(2)}, \cdots, \bldb^{(w_\bldaa)}$ (for all $\blda \in \code_\Gamma^{(\bldkk)}$) are well-defined. 
It is also straightforward to see that $\bldb^{(\ell)} \in \code_j$ for $\ell  = 1, 2, \cdots, w_\bldaa$. 
Next, we check that the newly-defined $w_{j,\bldbb}$ 
satisfy~(\ref{eq:equation-polytope-3})-(\ref{eq:equation-polytope-5}) for every $j \in \cJ, \; \bldb \in \code_j$. 

It is easy to see that $w_{j,\bldbb} \ge 0$; therefore~(\ref{eq:equation-polytope-3}) holds. 
By Proposition~\ref{prop:lemma-14} we obtain
\[
\sigma_{j,\bldkk} = \sum_{\bldbb \in \code_j, \; \bldsubkappa_j(\bldbb) = \bldkk} w_{j,\bldbb} \; ,  
\]
for all $j \in \cJ, \bldk \in \cT_j$, and
\[
\tau^{(\alpha)}_{i,j,\bldkk} = 
\sum_{\bldbb \in \code_j, \; \bldsubkappa_j(\bldbb) = \bldkk, \; b_i \cH_{j,i} = \alpha} w_{j,\bldbb} \; ,
\]
for all $j \in \cJ, \; i \in \cI_j, \; \bldk \in \cT_j, \; \alpha \in \rrrm$. 
Let $\beta \cH_{j,i} = \alpha$. Since
\[
\tau^{(\alpha)}_{i,j,\bldkk} = \sum_{\beta \; : \; \beta \cH_{j,i} = \alpha} z_{i,j,\bldkk}^{(\beta)} \; , 
\]
by the definition of $w_{j,\bldbb}$ it follows that  
\begin{eqnarray*}
\sum_{\bldbb \in \code_j, \; \bldsubkappa(\bldbb) = \bldkk, \; b_i = \beta} w_{j, \bldbb} & = & \frac{z_{i,j,\bldkk}^{(\beta)}}{\tau^{(\alpha)}_{i,j,\bldkk}} \quad \cdot \quad 
\sum_{\bldbb \in \code_j, \; \bldsubkappa(\bldbb) = \bldkk, \; b_i \cH_{j,i} = \alpha} w_{j, \bldbb} \\
& = & z_{i,j,\bldkk}^{(\beta)} \; , 
\end{eqnarray*}
where the first equality is due to the definition of the words $\bldb^{(\ell)}$, $\ell =1 ,2, \cdots, w_\bldaa$. 

By constraint~(\ref{eq:LP_2}) we have, for all $j \in \cJ$,
\begin{eqnarray*}
1 & = & \sum_{\bldkk \in \cT_j} \sigma_{j,\bldkk} \\ 
  & = & \sum_{\bldkk \in \cT_j} \quad \sum_{\bldbb \in \code_j, \; \bldsubkappa_j(\bldbb) = \bldkk} w_{j,\bldbb} \\
  & = & \sum_{\bldbb \in \code_j} w_{j,\bldbb} \; ,
\end{eqnarray*}
thus satisfying~(\ref{eq:equation-polytope-4}).

Finally, by constraint~(\ref{eq:LP_1}) we obtain, for all $j \in \cJ, i \in \cI_j, \beta \in \rrrm$,
\begin{eqnarray*}
f_i^{(\beta)} & = & \sum_{\bldkk \in \cT_j} z_{i,j,\bldkk}^{(\beta)} \\
   & = & \sum_{\bldkk \in \cT_j} \quad \sum_{\bldbb \in \code_j, \; \bldsubkappa_j(\bldbb) = \bldkk, \; b_i = \beta} w_{j,\bldbb} \\   
   & = & \sum_{\bldbb \in \code_j, \; b_i = \beta} w_{j,\bldbb} \; ,
\end{eqnarray*}
thus satisfying~(\ref{eq:equation-polytope-5}). 
\end{enumerate}
\end{proof}


\section{Cascaded Polytope Representation}
\label{sec:cascaded_polytope}

In this section we show that the ``cascaded polytope" representation described in~\cite{Chertkov},~\cite{Feldman-Yang} and~\cite{Feldman-Yang2} can be extended to nonbinary codes in a straightforward manner. Below, we elaborate on the details.

For $j \in \cJ$, consider the $j$-th row $\cH_j$ of the parity-check matrix 
$\cH$ over $\rrr$, and recall that 
\[
\code_j = \left\{ (b_i)_{i \in \cI_j} \; : \; \sum_{i \in \cI_j} b_i \cdot \cH_{j,i} = 0 \right\} \; .
\]
Assume that $\cI_j = \{ i_1, i_2, \cdots, i_{d_j} \}$ and denote $\cL_j = \{ 1, 2, \cdots, d_j-3 \}$.  
We introduce new variables 
\[
\bldchi^j = (\chi^j_i)_{i \in \cL_j} \; , 
\]
and denote 
\[
\bldchi = ( \bldchi^j )_{j \in \cJ} \; .  
\]

We define a new linear code $\code^{(\chi)}_j$ of length $2d_j-3$ by the
$(d_j - 2) \times (2d_j - 3)$ 
parity-check matrix $\cF_j$ associated with the following set of parity-check equations over $\rrr$:
\begin{enumerate}
\item
\begin{equation}
b_{i_1} \cH_{j,i_1} + b_{i_2} \cH_{j,i_2} + \chi^j_1 = 0 \; . 
\label{eq:kai-1}
\end{equation}
\item
For every $\ell = 1, 2, \cdots, d_j-4$, 
\begin{equation}
- \chi^j_\ell + b_{i_{\ell+2}} \cH_{j,i_{\ell+2}} + \chi^j_{\ell+1} = 0 \; . 
\label{eq:kai-2}
\end{equation}
\item
\begin{equation}
- \chi^j_{d_j-3} + b_{i_{d_j-1}} \cH_{j,i_{d_j-1}} + b_{i_{d_j}} \cH_{j,i_{d_j}} = 0 \; . 
\label{eq:kai-3}
\end{equation}
\end{enumerate}
We also define a linear code $\code^{(\chi)}$ of length $n + \sum_{j \in \cJ} (d_j-3)$ defined by the
$(\sum_{j \in \cJ} (d_j - 2)) \times (n + \sum_{j \in \cJ} (d_j-3))$ 
parity-check matrix $\cF$ associated with all the sets of parity-check equations~(\ref{eq:kai-1})-(\ref{eq:kai-3}) 
(for all $j \in \cJ$). We adopt the notation $\tilde{\bldb} = (\bldb \; | \;  \bldchi^j)$ for codewords of $\code^{(\chi)}_j$, and $\tilde{\bldc} = (\bldc \; | \;  \bldchi)$ for codewords of $\code^{(\chi)}$.

\begin{example}
Figure~\ref{cap:cascade-present} presents an example of the Tanner graph of a local code 
\[
\code_j = \left\{ (b_{i_1} \; b_{i_2} \; b_{i_3} \; b_{i_4} \; b_{i_5} \; b_{i_6}) \; : \; b_{i_1} + 2 b_{i_2} + 2 b_{i_3} + b_{i_4} + b_{i_5} + 2 b_{i_6} = 0 \right\} 
\]
of length $d_j = 6$ over $\rrr = \mathbb{Z}_3$, and the Tanner graph of the corresponding code $\code_j^{(\chi)}$ of length $9$ (three extra variables were added). The degree of every parity-check vertex in the Tanner graph of $\code_j^{(\chi)}$ is at most $3$. 
\begin{figure}
\begin{center}
\includegraphics[%
  width=0.7\columnwidth,
  keepaspectratio]{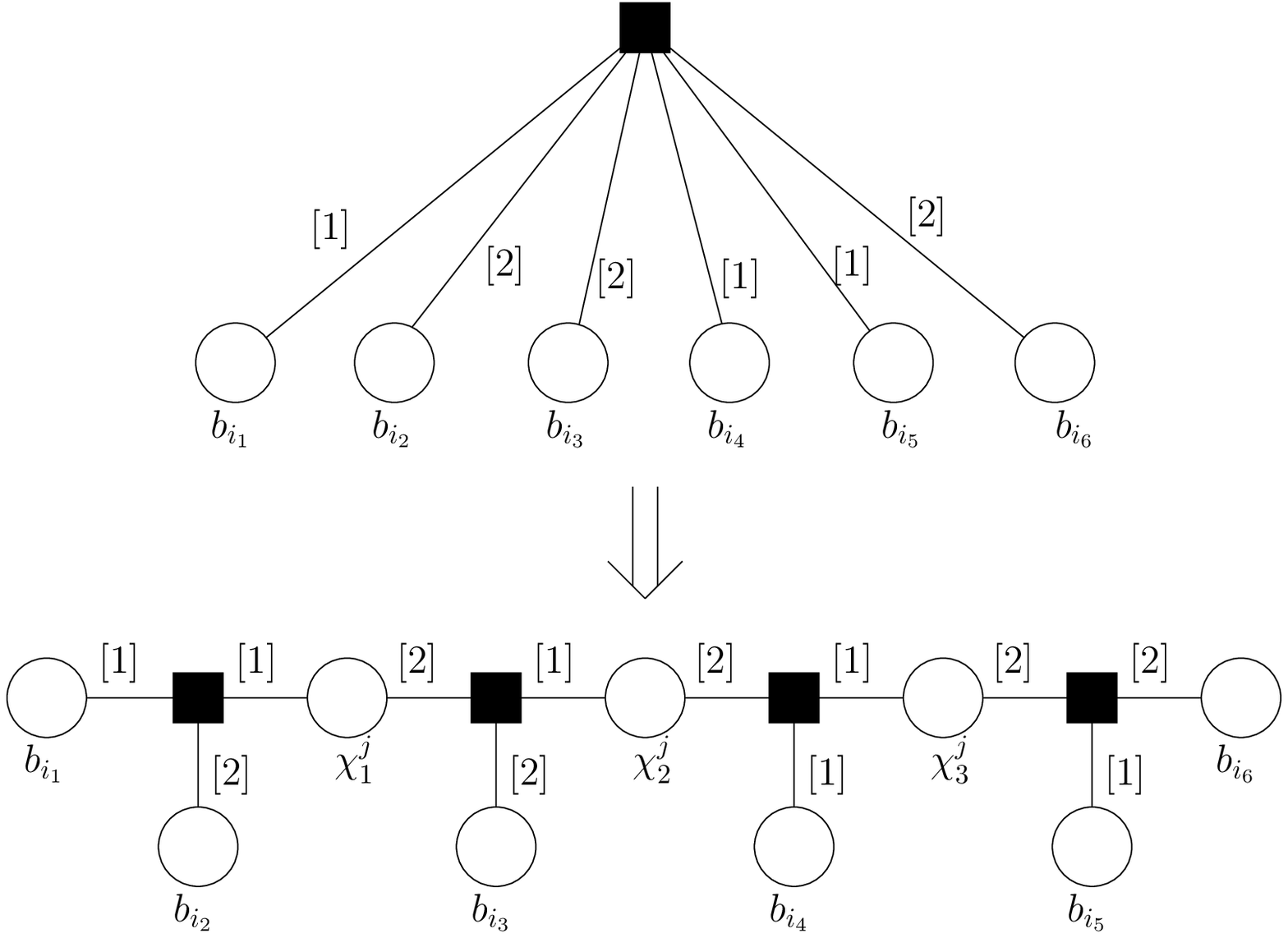}
\end{center}
\caption{\label{cap:cascade-present} Example of the Tanner graph of a local code $\code_j = \left\{ (b_{i_1} \; b_{i_2} \; b_{i_3} \; b_{i_4} \; b_{i_5} \; b_{i_6}) \; : \; b_{i_1} + 2 b_{i_2} + 2 b_{i_3} + b_{i_4} + b_{i_5} + 2 b_{i_6} = 0 \right\}$ of length $d_j = 6$ over $\rrr = \mathbb{Z}_3$, and its transformation into the Tanner graph of the corresponding code $\code_j^{(\chi)}$. Note that the degree of each parity-check vertex in the transformed graph is equal to $3$. }
\end{figure}
\end{example}

The following theorem relates the codes $\code_j$ and $\code^{(\chi)}_j$. 

\begin{theorem}
The vector $\bldb = (b_i)_{i \in \cI_j} \in \rrr^{d_j}$ is a codeword of $\code_j$ if and only if there exists a vector 
$\bldchi^j \in \rrr^{d_j-3}$ such that $(\bldb \; | \;  \bldchi^j) \in \code^{(\chi)}_j$. 
\label{thrm:code-equivalence}
\end{theorem} 

\begin{proof}
\begin{enumerate}
\item
Assume $\bldb = (b_i)_{i \in \cI_j} \in \code_j$. Define 
\begin{equation}
\hspace{-1ex}
\chi^j_\ell = \left\{ 
\begin{array}{cl} 
- b_{i_1} \cH_{j,i_1} - b_{i_2} \cH_{j,i_2} & \hspace{-1ex} \mbox{if } \ell = 1 \\
\chi^j_{\ell-1} - b_{i_{\ell+1}} \cH_{j,i_{\ell+1}} & \hspace{-1ex} 
\mbox{if } 2 \le \ell \le {d_j}-3 
\end{array} \right. \hspace{-2ex}
\label{eq:chi_as_fn_of_b}
\end{equation}
Then, obviously,~(\ref{eq:kai-1}) holds, and~(\ref{eq:kai-2}) holds for all $1 \le \ell \le {d_j}-4$. 
Finally,~(\ref{eq:kai-3}) follows from subtraction 
of~(\ref{eq:kai-1}) and~(\ref{eq:kai-2}) 
(for each $1 \le \ell \le {d_j}-4$) from the equation $\sum_{i \in \cI_j} b_i \cdot \cH_{j,i} = 0$. 
Therefore, $(\bldb \; | \;  \bldchi^j) \in \code^{(\chi)}_j$, as required. 
\item
Now, assume that $\bldb = (b_i)_{i \in \cI_j}$ is such that $(\bldb \; | \;  \bldchi^j) \in \code^{(\chi)}_j$ for some $\bldchi^j \in \rrr^{d_j-3}$, and thus~(\ref{eq:kai-1})--(\ref{eq:kai-3}) hold (in particular, (\ref{eq:kai-2}) holds
for all $1 \le \ell \le {d_j}-4$). 
We sum all the equalities in~(\ref{eq:kai-1})--(\ref{eq:kai-3}) and obtain that  
$\sum_{i \in \cI_j} b_i \cdot \cH_{j,i} = 0$. Therefore, $\bldb \in \code_j$. 
\end{enumerate}
\end{proof}

Note that from this theorem we may see that for every $\bldb \in \code_j$, there exists a unique 
$\bldchi^j = \bldchi^j(\bldb)$ such that $\tilde{\bldb} = (\bldb \; | \;  \bldchi^j) \in \code^{(\chi)}_j$, via (\ref{eq:chi_as_fn_of_b}); we may therefore use the notation $\tilde{\bldb}(\bldb) = (\bldb \; | \;  \bldchi^j(\bldb))$ to denote this unique completion, where $\bldchi^j(\bldb) = ( \chi_i^j(\bldb) )_{i \in \cL_j}$.

It follows from Theorem~\ref{thrm:code-equivalence} that 
the set of parity-check equations~(\ref{eq:kai-1})--(\ref{eq:kai-3}) for all $j \in \cJ$ 
equivalently describes the code $\code$. This description has at most $n + m \cdot (d-3)$ variables 
and $m \cdot (d-2)$ parity-check equations. However, the number of variables participating in
every parity-check equation is at most $3$. Therefore, the total number 
of variables and of constraints in the corresponding LP problem (defined by 
constraints~(\ref{eq:equation-polytope-3})-(\ref{eq:equation-polytope-5}) applied to the parity-check matrix $\cF$)  
is bounded from above by 
\[
(n + m (d-3))(q-1) + m(d-2) \cdot q^2  
\] 
and 
\[
m (d-2) (q^2 + 3q -2)  \; , 
\] 
respectively. 

In the sequel, we make use of some new notations, which we define next. 
First of all, with each parity-check equation prescribed by the matrix $\cF$, we associate a pair 
of indices $(j, \ell)$, $j \in \cJ$, $\ell = 1, 2, \cdots, d_j - 2$, where $j$ indicates the 
corresponding parity-check equation in $\cH$, and  $\ell$ indicates the serial number of the parity-check equation in the set of 
equations~(\ref{eq:kai-1})--(\ref{eq:kai-3}) corresponding to the $j$-th row of $\cH$. Denote by $\cI_{j,\ell} \subseteq \cI_j$ and $\cL_{j,\ell} \subseteq \cL_j$ 
the sets of indices $i$ of variables $b_i$ and $\chi_i^j$, respectively, 
corresponding to the non-zero entries in row $(j,\ell)$ of $\cF$. Then, each row of $\cF$ defines a single parity-check code $\code_{j,\ell}^{(\chi)}$. For any $\bldg \in \code_{j,\ell}^{(\chi)}$, we adopt the notation $\bldg = ( \bldg^b \; | \;  \bldg^{\chi} )$ where 
\[
\bldg^b = ( g_i^b )_{i \in \cI_{j,\ell}} \; ; \; \bldg^{\chi} = ( g_i^{\chi} )_{i \in \cL_{j,\ell}} \; .
\]

We denote by $\cS$ the polytope corresponding to the LP relaxation~(\ref{eq:equation-polytope-3})-(\ref{eq:equation-polytope-5}) for the code 
$\code^{(\chi)}$ with the parity-check matrix $\cF$. Recall that codewords of $\code^{(\chi)}$ are denoted $\tilde{\bldc} = (\bldc \; | \;  \bldchi)$. It is natural to represent points in $\cS$ as $((\bldf, \bldh), \bldz)$, 
where $\bldf = (f_i^{(\alpha)})_{i \in \cI, \; \alpha \in \rrrm}$ 
and $\bldh = (h_{j,i}^{(\alpha)})_{j \in \cJ, \; i \in \cL_j, \; \alpha \in \rrrm}$
are vectors of indicators corresponding to the entries $c_i$ $(i \in \cI)$ in $\bldc$ and 
$\chi^j_i$ $(j \in \cJ, \; i \in \cL_j)$ in $\bldchi$, respectively.
Here 
\[
\bldz = (z_{j, \ell, \bldgg})_{j \in \cJ, \; \ell = 1, 2, \cdots, d_j - 2, \; \bldgg \in \code^{(\chi)}_{j,\ell}}
\]
is a vector of weights associated with each parity-check equation 
$(j, \ell)$ and each codeword $\bldg \in \code^{(\chi)}_{j,\ell}$.

Similarly, for each $j \in \cJ$ we denote by $\cS_j$ the polytope corresponding to the LP relaxation~(\ref{eq:equation-polytope-3})-(\ref{eq:equation-polytope-5}) for the code $\code_j^{(\chi)}$, 
defined by the parity-check matrix $\cF_j$. Recall that codewords of $\code^{(\chi)}_j$ are denoted $\tilde{\bldb} = (\bldb \; | \;  \bldchi^j)$.
Then, it is also natural to represent points in $\cS_j$ as $((\hat{\bldf}_j, \hat{\bldh}_j), \hat{\bldz}_j)$, 
where $\hat{\bldf}_j = (f_i^{(\alpha)})_{i \in \cI_j, \; \alpha \in \rrrm}$ 
and $\hat{\bldh}_j = (h_{j,i}^{(\alpha)})_{i \in \cL_j, \; \alpha \in \rrrm}$ are vectors of indicators corresponding to the entries $b_i$ $(i \in \cI_j)$ in $\bldb$ and ${\chi^j_i}$ $(i \in \cL_j)$ in $\bldchi^j$, respectively.
Moreover, 
\[
\hat{\bldz}_j = (z_{j, \ell, \bldgg})_{\ell = 1, 2, \cdots, d_j - 2, \; \bldgg \in \code^{(\chi)}_{j,\ell}}
\]
is a vector of weights associated with each parity-check equation 
$(j, \ell)$ and each codeword $\bldg \in \code^{(\chi)}_{j,\ell}$. 

For each $j \in \cJ$, define the mapping $\bldXi_j$ analogously to the mapping $\bldXi$ with respect to the dimensionality of the code $\code_j^{(\chi)}$, namely 
\[
\bldXi_j \; : \; \rrr^{2d_j-3} \longrightarrow \{ 0, 1 \}^{(q-1)(2d_j-3)} \subset \mathbb{R}^{(q-1)(2d_j-3)} \; , 
\]
such that for $\tilde{\bldb} = (\bldb \; | \;  \bldchi^j) \in \code^{(\chi)}_j$,
\begin{equation*}
\bldXi_j(\tilde{\bldb}) = ( \bldxi(b_{i_1}) \; | \; \bldxi(b_{i_2}) \; | \; \cdots \; | \; \bldxi(b_{i_{d_j}}) 
\; | \; \bldxi(\chi^j_{1}) \; | \; \bldxi(\chi^j_{2}) \; | \; \cdots \; | \; \bldxi(\chi^j_{d_j-3}) ) \; .
\end{equation*}

The next lemma is similar to one of the claims of Proposition~10 in~\cite{KV-IEEE-IT}.  
\begin{lemma}
Let $\code$ be a code of length $n$ over $\rrr$ with parity-check matrix $\cH$, 
and let $\cQ(\cH)$ be the corresponding polytope of the LP relaxation, i.e. the set of points $(\bldf, \bldw)$ satisfying~(\ref{eq:equation-polytope-3})-(\ref{eq:equation-polytope-5}). Let $\bar{\cQ}(\cH)$ denote the projection of $\cQ$ onto the $\bldf$ variables, i.e. 
\[
\bar{\cQ}(\cH) = \{ \bldf : \exists \; \bldw \mbox{ s.t. } (\bldf, \bldw) \in \cQ \}
\]
Denote by $\ppp$ the set of \emph{normalized} graph-cover pseudocodeword vectors associated with $\cH$. 
Then, $\bar{\cQ}(\cH) = \overline{\ppp}$, where $\overline{\ppp}$ is the closure of ${\ppp}$
under the usual (Euclidean) metric in $\rr^{(q-1)n}$. 
\label{lemma:eq-polyts} 
\end{lemma}
\begin{proof}
Generally, the proof is similar to the proof of the relevant parts of Proposition~10 in~\cite{KV-IEEE-IT}. It is largely based on the equivalence between the set of graph-cover pseudocodewords and the set of LP pseudocodewords (Theorem~\ref{thm:PCW_equivalence} and Corollary~\ref{thm:PCW_vector_equivalence}). 
We avoid many technical details, and mention only the main ideas. The proof consists of proving two main claims.
\begin{enumerate}
\item
$\ppp \subseteq \bar{\cQ}(\cH)$. \\
Given any normalized graph-cover pseudocodeword vector $\bldf \in \ppp$, 
by Corollary~\ref{thm:PCW_vector_equivalence} there must exist 
$\bldw$ with $(\bldf, \bldw) \in \cQ(\cH)$. Therefore $\bldf \in \bar{\cQ}(\cH)$.
\item
If a point in $\bar{\cQ}(\cH)$ has all rational entries, then it must also be in $\ppp$. \\
The proof follows the lines of the proof of Lemma~56 in~\cite{KV-IEEE-IT}. Let $(\bldf, \bldw) \in \cQ(\cH)$ be a point such that all entries in $\bldf$ are rational. Then for all $j \in \cJ$, the vector $\hat{\bldf}_j = ( \bldf_i )_{i \in \cI_j}$ lies in the convex hull $\cK(\code_j)$. For convenience in what follows, denote the index set $\Psi = \{ 1,2, \cdots , (q-1)n+1 \}$. 
Using Carath\'{e}odory's Theorem~\cite[p. 10]{Barvinok}, for all $j \in \cJ$ we may write $\bldf = \bldmu^{(j)} \bldP^{(j)}$ where $\bldmu^{(j)} = (\mu_i^{(j)})_{i \in \Psi}$ is a row vector of length $|\Psi|$ whose elements sum to unity, and $\bldP^{(j)}$ is a $|\Psi| \times (|\Psi|-1)$ matrix such that for each $i \in \Psi$, the $i$-th row of $\bldP^{(j)}$, denoted $\bldp^{(j)}_i$, satisfies $\bldp^{(j)}_i = \bldXi(\bldc)$ for some $\bldc \in \rrr^n$ with $\bldx_j(\bldc) \in \code_j$. 
Therefore,  
\[
(\bldf \; 1) = \bldmu^{(j)} (\bldP^{(j)} \; \ones) \; , 
\]
where $\ones$ denotes a vector of length $|\Psi|$ all of whose entries are equal to $1$, is a $|\Psi| \times |\Psi|$ system; therefore by Cr\'{a}mer's rule the solution for $\bldmu^{(j)}$ has all rational entries (this argument applies for every $j \in\cJ$). Let $M$ denote a common denominator of all variables in vectors $\bldmu^{(j)}$, for $j \in \cJ$. Define $h_i^{(\alpha)} = M f_i^{(\alpha)} \in \mathbb{R}$ for each $i \in \cI$, $\alpha\in\rrrm$ (it is easy to see that these variables must be nonnegative integers). Also define $\delta_i^{(j)} = M \mu_i^{(j)}$ for each $i \in \Psi$, $j \in \cJ$, and $\blddelta^{(j)} = (\delta_i^{(j)})_{i \in \Psi}$. We then have
\begin{equation}
\bldh = \blddelta^{(j)} \bldP^{(j)} \; .
\label{eq:h_delta_relation}
\end{equation}
Next define, for all $j \in \cJ$, $\bldb \in \code_j$, 
\[
z_{j,\bldbb} = \sum_{i \in \Psi \; : \; \bldppp^{(j)}_i = \bldXiXi(\bldcc), \; \bldxx_j(\bldcc) = \bldbb} \delta_i^{(j)} \; . 
\]
By comparing appropriate entries in the vector equation (\ref{eq:h_delta_relation}), we obtain
\begin{eqnarray*}
&  \forall j \in \cJ, \; \forall i \in \cI_j, \; \forall \alpha\in\rrrm, \nonumber \\
& h_i^{(\alpha)} = 
\sum_{\bldbb \in \code_j, \; b_i=\alpha} z_{j,\bldbb} \; ,
\end{eqnarray*} 
and so $(\bldh, \bldz)$ is an LP pseudocodeword (the preceding equation yields (\ref{eq:LP_PCW_1}), and (\ref{eq:LP_PCW_2}) follows from the fact that the sum of the entries in $\blddelta^{(j)}$ is equal to $M$ for all $j \in\cJ$, these entries being nonnegative integers). So the construction of Theorem~\ref{thm:PCW_equivalence}, part (1), yields a corresponding graph-cover pseudocodeword with graph-cover pseudocodeword vector $\bldh$. Therefore the corresponding normalized graph-cover pseudocodeword vector is $\bldf$, and so we must have $\bldf \in \ppp$. 
\end{enumerate}
The claim of the lemma follows. 
\end{proof}

The following proposition is a counterpart of Lemma~28 in~\cite{KV-IEEE-IT}. 
\begin{proposition}
Let $\code$ be a code of length $n$ over $\rrr$ with parity-check matrix $\cH$. 
Assume that the Tanner graph represented by $\cH$ is a tree. 
Then, the projected polytope $\bar{\cQ}(\cH)$ of the corresponding LP relaxation problem is equal to $\cK(\code)$. 
\label{prop:eq-kappa}
\end{proposition} 
\begin{proof}
The proof follows the lines of the proof of Lemma~28 in~\cite{KV-IEEE-IT}. 
Let $\graph$ be the labeled Tanner graph of the code $\code$ corresponding to $\cH$. 
Let $\tilde{\graph}$ be an $M$-cover of $\graph$ for some positive integer $M$. 
Since $\graph$ is a tree, $\tilde{\graph}$ is a collection of $M$ labeled trees 
which are copies of $\graph$. Let $\tilde{\code}$ be a code defined by the
parity-check matrix corresponding to this $\tilde{\graph}$. We obtain that 
\begin{equation*}
\tilde{\code} = \Big\{ \bldx \in \rrr^{Mn} \; : \; (x_{1,m}, x_{2, m}, \cdots, x_{n,m}) \in \code \mbox{ for all } m = 1, 2, \cdots, M \Big\} \; . 
\end{equation*}
Then, it is easy to see that the set of \emph{normalized} graph-cover pseudocodeword vectors of $\cH$, $\ppp$, 
is equal to $\cK(\code) \cap \qq^{(q-1)n}$. 

To this end, we apply Lemma~\ref{lemma:eq-polyts} to see that
\[
\bar{\cQ}(\cH) = \overline{\ppp} = \overline{\cK(\code) \cap \qq^{(q-1)n} } = \cK(\code) \; , 
\]
as required. 
\end{proof}

By taking $\code = \code_j^{(\chi)}$ and $\cH = \cF_j$ so that $\cQ(\cH) = \cS_j$ (for $j \in \cJ$), we immediately obtain the following corollary: 
\begin{corollary} 
For $j \in \cJ$, let 
\[
\bar{\cS_j} = \{ (\hat{\bldf}_j, \hat{\bldh}_j) : \exists \; \hat{\bldz}_j \mbox{ s.t. } ((\hat{\bldf}_j, \hat{\bldh}_j), \hat{\bldz}_j) \in \cS_j \}
\]
Then $\bar{\cS_j} = \cK(\code_j^{(\chi)})$. 
\label{corollary:eq-K-S_j}
\end{corollary}

The proof of the next theorem requires the following definition. Let $\bldb \in \code_j$, and let $\bldg \in \code^{(\chi)}_{j,\ell}$. We say that $\bldg$ \emph{coincides with} $\bldb$, writing $\bldg \bowtie \bldb$, if and only if $g_i^b = b_i $ for all $i \in \cI_{j,\ell}$ and $g_i^{\chi} = \chi^j_i(\bldb)$ for all $i \in \cL_{j,\ell}$.

\begin{theorem}
The set \\
$\bar{\cS} = \{ \bldf : \exists \; \bldh, \bldz \mbox{ s.t. } ((\bldf, \bldh), \bldz) \in \cS \}$
is equal to the set $\bar{\cQ} = 
\{ \bldf : \exists \; \bldw \mbox{ s.t. } (\bldf, \bldw) \in \cQ \}$, and
therefore, optimizing the linear cost function~(\ref{eq:object-function}) over $\cS$ is equivalent to optimizing 
over $\cQ$. 
\end{theorem}

\begin{proof}
\begin{enumerate}
\item
Let $\bldf \in \bar{\cQ}$. Then, there exists $\bldw$ such that $(\bldf, \bldw) \in \cQ$. 
Therefore, 
\begin{equation}
\forall j \in \cJ, \; \forall i \in \cI_j, \; \forall \alpha \in \rrrm,
\qquad f_i^{(\alpha)} = 
\sum_{\bldbb \in \code_j, \; b_i=\alpha} w_{j,\bldbb} \; .
\label{eq:equation-polytope-5-mod} 
\end{equation}
In addition, the entries in $\bldw$ satisfy~(\ref{eq:equation-polytope-3}) 
and~(\ref{eq:equation-polytope-4}).  

We set the values of the variables $z_{j,\ell,\bldgg}$ as follows: 
\begin{equation*}
 \forall j \in \cJ, \; \forall \ell = 1, 2, \cdots, d_j-2, \; \forall \bldg \in \code^{(\chi)}_{j,\ell}, 
\qquad z_{j,\ell, \bldgg} = \sum_{ \bldbb \in \code_j, \bldgg \bowtie \bldbb} w_{j, \bldbb}
\; . 
\end{equation*}
So we have that
\begin{multline}
\forall j \in \cJ, \; \forall \ell = 1, 2, \cdots, d_j-2, 
\; \forall i \in \cI_{j,\ell}, \; \forall \alpha \in \rrrm,  \\\sum_{\bldgg \in \code^{(\chi)}_{j,\ell}, \; g_i^b=\alpha} z_{j,\ell,\bldgg} = \sum_{\bldbb \in \code_j, \; b_i=\alpha} w_{j,\bldbb} = f_i^{(\alpha)} \; ,
\label{eq:equation-polytope-5-mod-4} 
\end{multline}
using~(\ref{eq:equation-polytope-5-mod}), since $\cI_{j,\ell} \subseteq \cI_j$ for all $\ell = 1, 2, \cdots, d_j-2$. In addition, we define the variables $h_{j,i}^{(\alpha)}$ as follows. 
\begin{equation}
\forall j \in \cJ, \; \forall i \in \cL_j, \; \forall \alpha \in \rrrm,
\qquad h_{j,i}^{(\alpha)} = \sum_{\bldbb \in \code_j, \;  \chi^j_i(\bldbb)=\alpha} w_{j,\bldbb} \; .
\label{eq:equation-polytope-5-mod-3} 
\end{equation} 
Note that all variables $h_{j,i}^{(\alpha)}$ are well defined. It then follows that
\begin{multline}
\forall j \in \cJ, \; \forall \ell = 1, 2, \cdots, d_j-2, \; 
\forall i \in \cL_{j, \ell}, \; \forall \alpha \in \rrrm, \\ 
\sum_{\bldgg \in \code^{(\chi)}_{j,\ell}, \;  g_i^\chi=\alpha} z_{j,\ell,\bldgg} = \sum_{\bldbb \in \code_j, \;  \chi^j_i(\bldbb)=\alpha} w_{j,\bldbb} = h_{j,i}^{(\alpha)} \; ,
\label{eq:equation-polytope-5-mod-5} 
\end{multline} 
using~(\ref{eq:equation-polytope-5-mod}), since $\cL_{j,\ell} \subseteq \cL_j$ for all $\ell = 1, 2, \cdots, d_j-2$. 

Next, we claim that 
\begin{equation} 
((\bldf, \bldh), \bldz) \in \cS \; . 
\label{eq:in-S}
\end{equation}
In order to show this, it is necessary to show~(\ref{eq:equation-polytope-3})-(\ref{eq:equation-polytope-5}) with 
respect to $((\bldf, \bldh), \bldz)$ and the code $\code^{(\chi)}$. 
However~(\ref{eq:equation-polytope-3}) and~(\ref{eq:equation-polytope-4})
follow easily from the definition of the variables $z_{j,\ell,\bldgg}$ and the properties of 
the variables $w_{j,{\bldbb}}$. As to~(\ref{eq:equation-polytope-5}), it follows from the 
combination of~(\ref{eq:equation-polytope-5-mod-4}) and (\ref{eq:equation-polytope-5-mod-5}).  

Finally,~(\ref{eq:in-S}) yields that $\bldf \in \bar{\cS}$, as required. 

\item
Now, assume that $\bldf \in \bar{\cS}$. This means that there exist $\bldh$, $\bldz$ such that $((\bldf, \bldh), \bldz) \in \cS$. 
Then, for all $j \in \cJ$, $((\hat{\bldf}_j, \hat{\bldh}_j), \hat{\bldz}_j) \in \cS_j$.
By Corollary~\ref{corollary:eq-K-S_j}, $(\hat{\bldf}_j, \hat{\bldh}_j)$ lies in $\cK(\code_j^{(\chi)})$. 
Therefore, 
\begin{equation}
(\hat{\bldf}_j, \hat{\bldh}_j) = \sum_{\tilde{\bldbb} \in \code_j^{(\chi)}} \beta_{j, \tilde{\bldbb}} \cdot \bldXi_j (\tilde{\bldb}), 
\label{eq:f-convex-hull}
\end{equation}
where $\sum_{\tilde{\bldbb} \in \code_j^{(\chi)} } \beta_{j, \tilde{\bldbb}} = 1$ and $\beta_{j, \tilde{\bldbb}} \ge 0$ for all 
$\tilde{\bldb} \in \code_j^{(\chi)}$. 

For all $j \in \cJ$, ${\bldb} \in \code_j$ set the value of $w_{j,{\bldbb}}$ as 
\begin{equation*}
 w_{j, \bldbb} = \beta_{j,\tilde{\bldbb}(\bldbb)} \; ,
\end{equation*}
and thus
\begin{equation}
\sum_{{\bldbb} \in \code_j } w_{j, {\bldbb}} = 1 \; , 
\label{eq:w-1}
\end{equation}
and 
\begin{equation}
w_{j, {\bldbb}} \ge 0 \quad \mbox{ for all } {\bldb} \in \code_j \; .
\label{eq:w-2}
\end{equation}
Then,~(\ref{eq:f-convex-hull}) becomes 
\begin{equation*}
(\hat{\bldf}_j, \hat{\bldh}_j) = \sum_{{\bldbb} \in \code_j} w_{j, {\bldbb}} \cdot \bldXi_j (\tilde{\bldb}(\bldb)) \; .
\end{equation*}
Comparing the first set of coordinates, we obtain that 
\begin{equation*}
\forall i \in \cI_j, \; \forall \alpha \in \rrrm, \quad
f_i^{(\alpha)} = 
\sum_{{\bldbb} \in \code_j, \; b_i=\alpha} w_{j, {\bldbb}} \; .
\end{equation*}
This set of equations holds for all $j \in \cJ$. 
Together with~(\ref{eq:w-1}) and~(\ref{eq:w-2}) this means that $(\bldf, \bldw) \in \cQ$. 
Therefore, $\bldf \in \bar{\cQ}$, as required.  
\end{enumerate}
\end{proof}

The polytope representation described in this section leads to a polynomial-time 
decoder for a wide variety of classical nonbinary codes (for example, generalized Reed-Solomon codes). 

\section{Simulation Study}
\label{sec:sims}
\subsection{Comparison with ML Decoding}
In this section we compare performance of the linear-programming decoder with hard-decision and soft-decision based ML decoding. For such a comparison, a code and modulation scheme are needed which possess sufficient symmetry properties to enable derivation of analytical ML performance results. We consider encoding of $6$-symbol blocks
according to the $\left[11,6\right]$ ternary Golay code, and modulation
of the resulting ternary symbols with $3$-PSK modulation prior to
transmission over the AWGN channel. Figure~\ref{cap:Golay} shows the symbol error rate (SER) and codeword error rate (WER) performance of this code under LP decoding using the polytope $\cQ$ of Section \ref{sec:lp}. Note that this is the same as its performance using the polytope $\cU$ of Section \ref{sec:alternative_polytope}, and its performance using the polytope $\cS$ of Section \ref{sec:cascaded_polytope}. When the decoder reports a decoding failure, the SER and WER are both taken to be $1$. To
quantify performance, we define the signal-to-noise ratio (SNR) per
information symbol $\gamma_{s}=E_{s}/N_{0}$ as the ratio of the received
signal energy per information symbol to the noise power spectral density.
Also shown in the figure are two other performance curves for WER.
\begin{figure}
\begin{center}\includegraphics[%
  width=0.9\columnwidth,
  keepaspectratio]{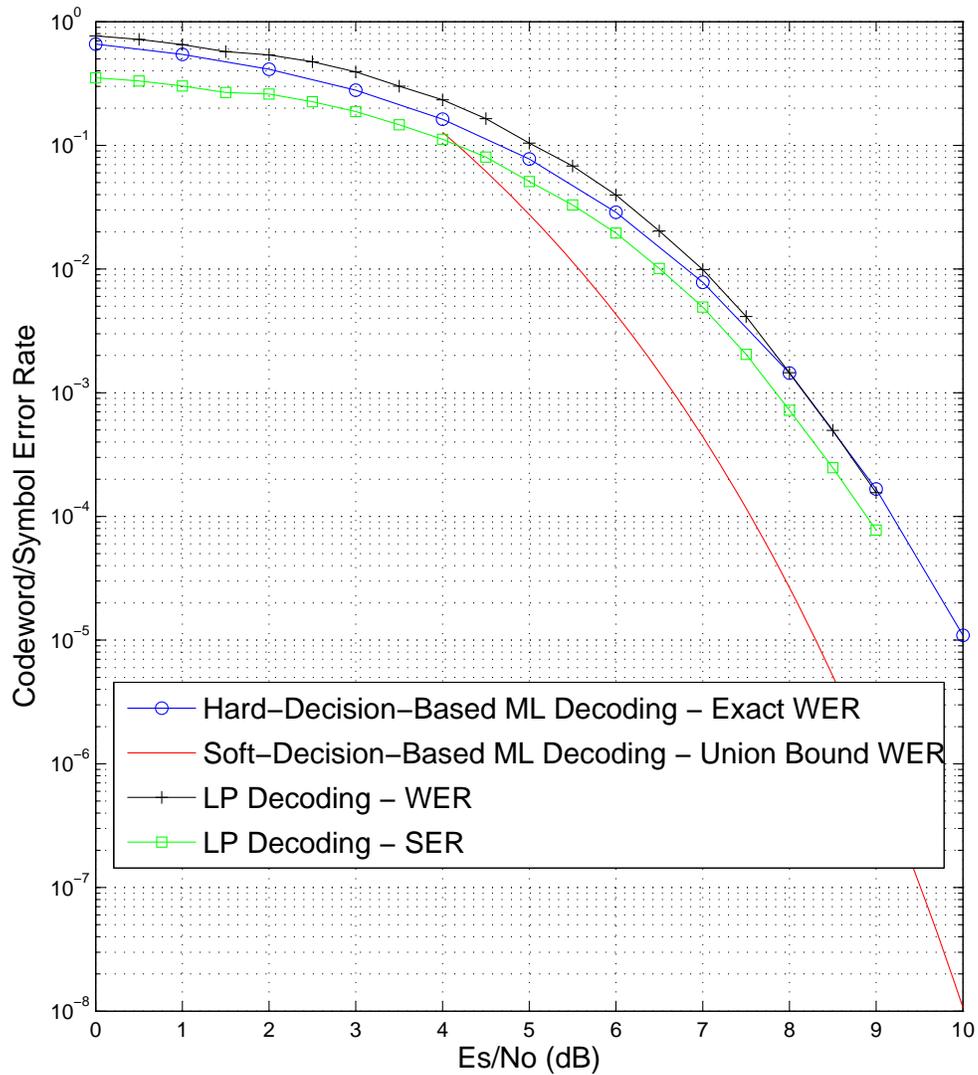}
\end{center}
\caption{\label{cap:Golay}Codeword error rate (WER) and symbol error rate (SER) for the $[11,6]$ ternary
Golay code with $3$-PSK modulation over the AWGN channel. The figure shows performance
under LP decoding, as well as the exact result for hard-decision decoding
and the union bound for soft-decision decoding.}
\end{figure}
The first is the exact result for ML hard-decision decoding of the
ternary Golay code; since the Golay code is perfect, this is obtained
from 
\[
\textrm{WER}(\gamma_s)=\sum_{\ell=3}^{11}\binom{11}{\ell}(p(\gamma_s))^{\ell}\left(1-p(\gamma_s)\right)^{11-\ell} \; , 
\]
where $p(\gamma_s)$ represents the probability of incorrect hard decision at
the demodulator and was evaluated for each value of $\gamma_s$ using numerical
integration. The second WER curve represents the union bound for ML
soft-decision decoding. Using the symmetry of the $3$-PSK constellation,
this may be obtained from 
\[
\textrm{WER}(\gamma_s)<\frac{1}{2}\sum_{\bldcc \in \code}
\textrm{erfc }\left(\sqrt{\frac{3}{4}w_{H}(\bldc) R(\code) \gamma_{s}}\right) \; , 
\] 
where $R(\code) = 6/11$ denotes the code rate, and the Hamming weight of the codeword $\bldc \in \code$, $w_{H}(\bldc)$, 
is distributed according to the weight enumerating polynomial~\cite{MacWilliams_Sloane}
\[
W\left(x\right)=1+132x^{5}+132x^{6}+330x^{8}+110x^{9}+24x^{11} \; . 
\]
The performance of LP decoding is approximately the same as that of
codeword-error-rate optimum hard-decision decoding. The performance
lies $0.1$ dB from the result for ML hard-decision decoding and
$1.53$ dB from the union bound for codeword-error-rate optimum soft-decision
decoding at a WER of $10^{-4}$. These results are comparable to those of a similar study conducted for the binary case in \cite{Feldman}.
\subsection{Low-Density Code Performance}
Figure \ref{cap:cyclic_LDPC} shows SER and WER simulation performance results for two low-density parity-check (LDPC) codes. The first code $\code^{(1)}$, of length $n=150$, is over the ring $\rrr=\mathbb{Z}_3$, where nonbinary coded symbols are mapped directly to ternary PSK signals and transmitted over an AWGN channel, the mapping described in Example \ref{ex:qary_PSK_AWGN} being used for modulation. The parity-check matrix $\cH^{(1)}$ consists of $m=60$ rows and is equal to the right-circulant matrix
\[
\cH^{(1)}_{j,i} = \left\{ \begin{array}{cc}
1 & \textrm{ if } i - j \in \{ 0, 51, 80 \} \\
2 & \textrm{ if } i - j \in \{ 8, 30, 90 \} \\
0 & \textrm{ otherwise. }\end{array}\right. \;
\]
\begin{figure}
\begin{center}\includegraphics[%
  width=0.9\columnwidth,
  keepaspectratio]{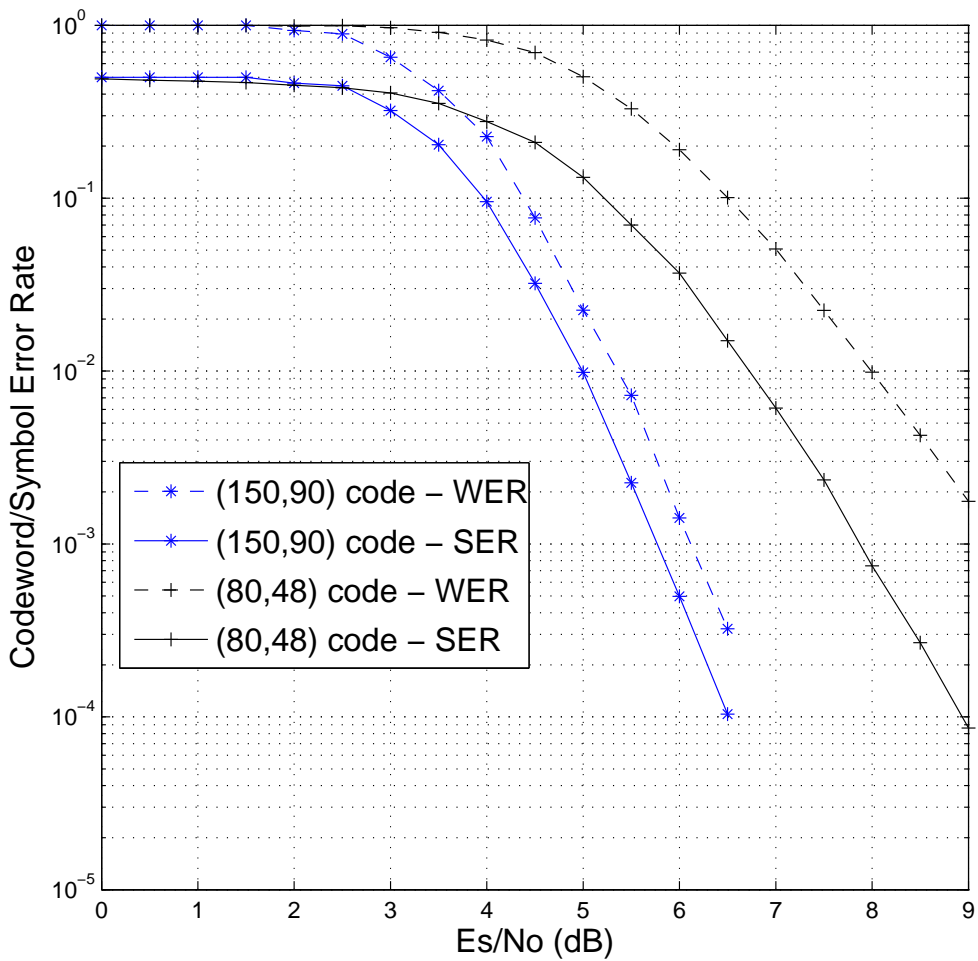}
\end{center}
\caption{\label{cap:cyclic_LDPC} Codeword error rate (WER) and symbol error rate (SER) for the $[150,90]$ ternary
LDPC code $\code^{(1)}$ under ternary PSK modulation, and for the $[80,48]$ quaternary LDPC code $\code^{(2)}$ under QPSK modulation.}
\end{figure}
The code rate is $R(\code^{(1)}) = 0.6$. As expected, the performance of the low-density code $\code^{(1)}$ is significantly better than that of the ternary Golay code given in Figure~\ref{cap:Golay}. The second code $\code^{(2)}$, of length $n=80$, is over the ring $\rrr=\mathbb{Z}_4$, where nonbinary coded symbols are mapped directly to quaternary phase shift keying (QPSK) signals and transmitted over an AWGN channel, the mapping described in Example \ref{ex:qary_PSK_AWGN} again being used for modulation. The parity-check matrix $\cH^{(2)}$ consists of $m=32$ rows and is equal to the right-circulant matrix
\[
\cH^{(2)}_{j,i} = \left\{ \begin{array}{cc}
1 & \textrm{ if } i - j \in \{ 0, 41, 48 \} \\
3 & \textrm{ if } i - j \in \{ 8, 25 \} \\
0 & \textrm{ otherwise. }\end{array}\right. \;
\]
This code also has rate $R(\code^{(2)}) = 0.6$. The quaternary code has a higher SER and WER than the ternary code for the same $E_s/N_0$; however it has a smaller block length and a higher spectral efficiency. In both systems, when the decoder reports a decoding failure the SER and WER are both taken to be $1$.

\section{Future Research}
\label{sec:future}

Sections~\ref{sec:alternative_polytope} and~\ref{sec:cascaded_polytope} presented two alternative polytope representations, which have a smaller number of variables and constraints
than the respective standard LP representation in certain contexts. It would be interesting to further reduce 
the complexity of the polytope representation in order to yield more efficient decoding algorithms. Alternatively, one could try to reduce complexity of the LP solver for the nonbinary decoding problem
by exploiting knowledge of the polytope structure. 

The notion of \emph{pseudodistance} for nonbinary codes was recently defined in~\cite{SF}, and lower bounds on the pseudodistance of nonbinary codes under $q$-ary PSK modulation 
over the AWGN channel were presented. It would be interesting to obtain lower bounds on the pseudodistance for other families of nonbinary linear codes and for other modulation schemes.

\appendix

\subsection*{Proof of Proposition~\ref{prop:lemma-14}}

Preliminary to proving this Proposition we give some background material on flow networks. 

\subsubsection*{Flow Networks} 
 
Let $\sG = (\sV, \sE)$ be a directed graph, and let $\{ s,t \} \subseteq \sV, s \neq t$.  
A flow network $(\sG(\sV, \sE), \scp)$ is a graph $\sG = (\sV, \sE)$ with 
a nonnegative \emph{capacity} function $\scp : \sE \longrightarrow \rr \cup \{+\infty\}$ defined
for every edge.  

For a subset $\sV' \subseteq \sV$ let $\sV'' = \sV \backslash \sV'$. We 
define a cut $(\sV' : \sV'')$ induced by $\sV'$ as a set of edges 
$ \{ (u,v) \; : \; u \in \sV', \; v \in \sV'' \}$.  
The capacity of this cut, $\scp(\sV' : \sV'')$,  is defined as 
\[
\scp(\sV' : \sV'') = \sum_{u \in \sV', \; v \in \sV''} \scp((u,v)) \; . 
\]

For the edge $e = (u, v)$ we use the notation $e \in \mbox{in}(v)$ and $e \in \mbox{out}(u)$. 
We also use the notation $\ngbr(v)$ to denote the set of neighbors of $v$, namely 
\[
\ngbr(v) = \left\{ u \; : \; (u,v) \in \sE \right\} \cup \left\{ u' \; : \; (v,u') \in \sE \right\} \; . 
\]
For a set of vertices $\sV_0 \subseteq \sV$, denote 
\[
\ngbr(\sV_0) = \cup_{v \in \sV_0} \ngbr(v) \backslash \sV_0 \; .
\]  
The flow in the graph (network) $\sG$ with a source $s$ and a sink $t$ 
is defined as a function $\sfn : \sE \longrightarrow \rr \cup \{ + \infty \}$
that satisfies $0 \le \sfn(e) \le \scp(e)$ for all $e \in \sE$, and 
\[
\forall v \in \sV \backslash \{s, t \}, \; \sum_{e \in \sE, \; e \in \mbox{\scriptsize in}(v)} \sfn(e) = 
\sum_{e \in \sE, \; e \in \mbox{\scriptsize out}(v)} \sfn(e) \; . 
\]
The value of the flow $\sfn$ is defined as 
\[
\sum_{e \in \sE, \; e \in \mbox{\scriptsize in}(t)} \sfn(e) = 
\sum_{e \in \sE, \; e \in \mbox{\scriptsize out}(s)} \sfn(e) \; . 
\]

The maximum flow in the network is defined as the flow $\sfn$ that attains the maximum 
possible value. There are several known algorithms, for instance the Ford-Fulkerson algorithm, 
for finding the maximum flow in a network, 
the reader can refer to~\cite[Section 26.2]{Cormen}. It is well known that the value of
the maximum flow in the network is equal to the capacity of the minimum cut induced by a vertex set 
$\sV'$ such that $s \in \sV'$ and $t \notin \sV'$ (see~\cite{Cormen}). 

Finally, we prove the Proposition. 

\begin{proof}
The proof will be by induction on $M$. 
We set $w_\bldaa = 0$ for all $\blda \in \code_\Gamma^{(\bldkk)}$. 
We show that there exists a vector $\blda = \left( a_i \right)_{i \in \Gamma} \in \code_\Gamma^{(\bldkk)}$
such that 
\begin{enumerate}
\item[(i)]
For every $i \in \Gamma$ and $\alpha \in \rrrm$, 
\[
a_i = \alpha \quad \Longrightarrow \quad x^{(\alpha)}_i > 0 \; . 
\]
\item[(ii)]
If for some $i \in \Gamma$, $\sum_{\alpha \in \rrrm} x^{(\alpha)}_i = M$, then 
$a_i = \alpha$ for some $\alpha \in \rrrm$. 
\end{enumerate}

Then, we `update' the values of $x^{(\alpha)}_i$'s and $M$ as follows. For 
every $i \in \Gamma$ and $\alpha \in \rrrm$ with 
$a_i = \alpha$ we set $x^{(\alpha)}_i \leftarrow x^{(\alpha)}_i - 1$. 
In addition, we set $M \leftarrow M - 1$. We also set $w_\bldaa \leftarrow w_\bldaa + 1$. 

It is easy to see that the `updated' values of $x^{(\alpha)}_i$'s and $M$ satisfy  
\[
\sum_{i \in \Gamma} x^{(\alpha)}_i = k_\alpha M
\]
 for all $\alpha \in \rrrm$, 
and $\sum_{\alpha \in \rrrm} x^{(\alpha)}_i \le M$ for all $i \in \Gamma$.
Therefore, the inductive step can be applied with respect to these new values. 
The induction ends when the value of $M$ is equal to zero. 

It is straightforward to see that when the induction 
ter\-mi\-na\-tes, (\ref{eq:lemma14-claim-1}) and~(\ref{eq:lemma14-claim-2}) hold with respect 
to the origi\-nal va\-lu\-es of the $x^{(\alpha)}_i$ and $M$. 
\newline

\subsubsection*{Proof of existence of $\blda$ that satisfies (i)}

We construct a flow network  $\sG = (\sV, \sE)$ as follows: 
\[
\sV = \{ s, t \} \cup \sU_1 \cup \sU_2 \; , 
\]
where 
\[
\sU_1 = \rrrm \quad \mbox{ and } \quad
\sU_2 = \Gamma \; . 
\]
Also set
\[
\sE = \{ (s, \alpha) \}_{\alpha \in \rrrm} \; \cup 
\; \{ (i, t) \}_{i  \in \Gamma} 
\; \cup \; \{ (\alpha, i) \}_{x^{(\alpha)}_i > 0} \; .  
\]
We define an integer capacity function $\scp : \sE \longrightarrow \nn \cup \{ + \infty \}$ as follows:
\begin{equation}
\scp(e) = \left\{ 
\begin{array}{cl}
k_\alpha & \mbox{ if } e = (s, \alpha), \; \alpha \in \rrrm \\
1 & \mbox{ if } e = (i, t), \; i \in \Gamma \\
+\infty & \mbox{ if } e = (\alpha, i), \; \alpha \in \rrrm, \; i \in \Gamma
\end{array} \right. \; . 
\end{equation}

Next, apply the Ford-Fulkerson algorithm on the network $(\sG(\sE, \sV), \scp)$ to produce a maximal flow $\sfn_{max}$. Since 
all the values of $\scp(e)$ are integer for all $e \in \sE$, so 
the values of 
$\sfn_{max}(e)$ must all be integer for every $e \in \sE$ (see~\cite{Cormen}).

We will show that the minimum cut in this graph has capacity $\scp_{min} = \sum_{\alpha \in \rrrm} k_\alpha$. 
First, consider the cut induced by the set $\sV' = \{ s \}$. This cut has capacity 
$\sum_{\alpha \in \rrrm} k_\alpha$, and therefore $\scp_{min} \le \sum_{\alpha \in \rrrm} k_\alpha$. 

Assume that there is another cut, which has smaller capacity. If this smaller cut is induced by the set 
$\sV' = \sV \backslash \{ t \}$, its capacity must be $N \ge \sum_{\alpha \in \rrrm} k_\alpha$ $-$
it is not smaller. Therefore, without loss of generality, assume that the minimum cut is 
induced by the set $\sV'$, where $\sV' = \{ s \} \; \cup \sX' \; \cup \; \sY'$, 
$\sX' \subseteq \sU_1$ and $\sY' \subseteq \sU_2$. Let $\sX'' = \sU_1 \backslash \sX'$ and 
$\sY'' = \sU_2 \backslash \sY'$ (and so $\sV'' = \{ t \} \; \cup \sX'' \; \cup \; \sY''$).

Observe that there are 
no edges $(\alpha, i) \in \sE$ with $\alpha \in \sX'$, 
$i \in \sY''$, because otherwise the capacity of the respective cut would be infinitely large
(so it cannot be a minimum cut). Thus, 
\begin{equation}
\label{eq:bound-Y-1}
|\sY'| \ge \left| \sU_2 \cap \ngbr(\sX') \right|
\end{equation}

Observe also that 
\[
\sum_{i \in \Gamma} \sum_{\alpha \in \sX'} x^{(\alpha)}_i = 
\sum_{\alpha \in \sX'} k_\alpha M 
\] 
and 
\[
\sum_{\alpha \in \sX'} x^{(\alpha)}_i \le \sum_{\alpha \in \rrrm} x^{(\alpha)}_i \le M \; . 
\]
Therefore, 
\begin{equation}
\left| \sU_2 \cap \ngbr(\sX') \right| \ge 
\sum_{\alpha \in \sX'} k_\alpha \; .
\label{eq:bound-Y-2}
\end{equation}
We obtain that 
\begin{equation}
\scp(\sV':\sV'') = \sum_{ \alpha \in \sX''} k_\alpha + |\sY'|
\ge \sum_{ \alpha \in \sX''} k_\alpha + \sum_{ \alpha \in \sX'} k_\alpha 
= \sum_{\alpha \in \rrrm} k_\alpha \; , 
\end{equation}
where the inequality is due to~(\ref{eq:bound-Y-1}) and~(\ref{eq:bound-Y-2}).  
This leads to a contradiction of the non-minimality of $\scp (\sV' : \sV'')$ for $\sV' = \{ s \}$.  

If we apply the Ford-Fulkerson algorithm (or a similar algorithm) on the network $(\sG(\sV,\sE), \scp)$, we obtain 
that the integer flow $\sfn_{max}$ in $\sG$ has a value of $\sum_{\alpha \in \rrrm} k_\alpha$. 
Observe that $\sfn_{max}((\alpha, i)) \in \{ 0, 1 \}$ for all 
$\alpha \in \rrrm$ and $i \in \Gamma$. 
Then, for all $i \in \Gamma$, we define 
\[
a_i = \left\{ \begin{array}{cl}
\alpha & \mbox{ if } \sfn_{max}((\alpha, i)) = 1 \mbox{ for some } \alpha \in \sU_1 \\ 
0 & \mbox{ otherwise } 
\end{array} \right. \; . 
\]
For this selection of $\blda = (a_1, a_2, \cdots, a_N)$, we have $\blda \in \code_\Gamma^{(\bldkk)}$
and $a_i = \alpha$ only if $x^{(\alpha)}_i > 0$. 
\newline

\subsubsection*{Proof of existence of $\blda$ that satisfies (i) and (ii) simultaneously} 

We start with the following definition. 
\begin{definition}
The vertex $i \in \sU_2$ is called a \emph{critical} vertex, if 
\[
     \sum_{\alpha \in \rrrm} x^{(\alpha)}_i = M \; . 
\]
\end{definition}
In order to have~(\ref{eq:lemma-14-req-2}) satisfied after the next inductive step, we have 
to decrease the value of $\sum_{\alpha \in \rrrm} x^{(\alpha)}_i$ by (exactly) 1 
for every critical vertex. This is equivalent to having $\sfn_{max}((i, t)) = 1$. 

We have just shown that the maximum (integer) flow in $\sG$ has value $\sum_{\alpha \in \rrrm} k_\alpha$. 
Now, we aim to show that there exists a flow $\sfn^*$ of the same value, which has 
$\sfn^*((i, t)) = 1$ for every critical vertex $i$. 

Suppose that 
there is no such flow. Then, consider the maximum flow $\sfn'$, which has 
$\sfn'((i, t)) = 1$ for the \emph{maximal possible number} of the critical vertices $i \in \sU_2$. 
In the sequel, we assume that there is a critical vertex ${i_0} \in \sU_2$, 
which has $\sfn'(({i_0}, t)) = 0$. 
We will show that the flow $\sfn'$ can be modified towards the flow $\sfn''$ of the same value, 
such that for $\sfn''$ the number of critical vertices $i \in \sU_2$ having 
$\sfn''((i, t)) = 1$ is strictly larger than for $\sfn'$. 

Indeed, if there exists vertex ${\alpha_0} \in \ngbr(i)$ 
such that $({\alpha_0}, {i_1}) \in \sE$ and 
$\sfn'(({\alpha_0}, {i_1})) = 1$ for some \emph{non-critical} vertex ${i_1}$, then 
$\sfn'(({\alpha_0}, {i_0})) = 0$, $\sfn'(({i_0}, t)) = 0$ 
and $\sfn'(({i_1}, t)) = 1$. We define the flow $\sfn''$ as
\[
\sfn''(e) = \left\{ \begin{array}{cl}
1 & \mbox{ if } e \in \{ ({\alpha_0}, {i_0}), ({i_0}, t) \} \\
0 & \mbox{ if } e \in \{ ({\alpha_0}, {i_1}), ({i_1}, t) \} \\
\sfn'(e) & \mbox{ for all other edges $e \in \sE$}
\end{array} \right. \; . 
\]
It is easy to see that $\sfn''$ is a legal flow in $(\sG(\sV,\sE), \scp)$. Moreover, it has the same value as $\sfn'$, and the 
number of critical vertices $i \in \sU_2$ satisfying 
$\sfn''((i, t)) = 1$ is strictly larger than for $\sfn'$. 

In the general case (when there is no vertex ${\alpha_0}$ as above), we iteratively define a 
maximal set $\sZ$ of vertices $\alpha \in \sU_1$ satisfying the next two rules:
\begin{enumerate}
\item
For any ${\alpha} \in \sU_1$: if $({\alpha}, {i_0}) \in \sE$ then ${\alpha} \in \sZ$.  
\item
For any ${\alpha} \in \sU_1$ and $i \in \sU_2$: 
if $({\alpha}, {i}) \in \sE$, $\sfn'(({\alpha}, {i})) = 0$,
and there exists ${\beta} \in \sZ$ such that 
$({\beta}, {i}) \in \sE$, $\sfn'(({\beta}, {i})) = 1$ and all ${i'} \in \sU_2$ with 
$\sfn'(({\beta}, {i'})) = 1$ are critical, then ${\alpha} \in \sZ$.    
\end{enumerate} 

Consider the set $\sZ$. There are two cases. 
\begin{description}
\item[{\sl Case 1:}] $\quad$
Every vertex ${\alpha}$ in $\sZ$ satisfies 
\begin{equation*}
\forall i \in \sU_2 : ({\alpha}, i) \in \sE \mbox{ and } 
{i} \mbox{ is not critical }
\Longrightarrow \quad \sfn'(({\alpha}, {i})) = 0 \; . 
\end{equation*} 
Then, for every ${\alpha} \in \sZ$ there are exactly $k_\alpha$ vertices $i$ such that 
$(\alpha, i) \in \sE$ and $\sfn'(({\alpha}, {i})) = 1$. Define 
\begin{equation*}
\sT = \{ i \in \sU_2 \mbox{ is critical } 
\; : \; \exists \alpha \in \sZ \mbox{ s.t. }
(\alpha, i) \in \sE \; \mbox{ and } \; \sfn'(({\alpha}, {i})) = 1  \} \; . 
\end{equation*}
We have 
\begin{equation}
| \sT | = \sum_{i \in \sT} 1 = \sum_{\alpha \in \sZ} k_\alpha \; . 
\label{eq:k-alpha-vertices}
\end{equation}

Note that $i_0 \notin \sT$ and recall that 
\begin{equation}
(\beta, {i_0}) \in \sE \mbox{ for some }  {\beta} \in \sZ \; .
\label{eq:edge-rho-i-0}
\end{equation} 

Note also that if $\gamma \notin \sZ$ and $i \in \sT$, 
then there is no edge between $\gamma$ and $i$
(otherwise, $\sfn'((\gamma, i)) = 0$, and so $\gamma$ should be in $\sZ$). 
Therefore, $x^{(\gamma)}_i = 0$, and so 
\begin{equation}
\sum_{\alpha \in \sZ} \sum_{i \in \sT} x^{(\alpha)}_i 
= \sum_{\alpha \in \rrrm} \sum_{i \in \sT} x^{(\alpha)}_i \; . 
\label{eq:no-edges-equation}
\end{equation}

We obtain
\begin{eqnarray*}
\sum_{\alpha \in \sZ} k_\alpha M & = & 
\sum_{\alpha \in \sZ} \sum_{i \in \Gamma} x^{(\alpha)}_i > 
\sum_{\alpha \in \sZ} \sum_{i \in \sT} x^{(\alpha)}_i \\
& = & \sum_{\alpha \in \rrrm} \sum_{i \in \sT} x^{(\alpha)}_i 
= \sum_{i \in \sT} \sum_{\alpha \in \rrrm} x^{(\alpha)}_i \\
& = & \sum_{i \in \sT} M = \sum_{\alpha \in \sZ} k_\alpha M \; . 
\end{eqnarray*}
Here the first equality is due to~(\ref{eq:lemma-14-req-1}), 
the strict inequality is due to~(\ref{eq:edge-rho-i-0}) and the second equality is due to~(\ref{eq:no-edges-equation}). 
The third equality is obtained by
the change of the order of the summation. The fourth equality is true because all vertices in $\sT$ are critical. 
Finally, the fifth equality is due to~(\ref{eq:k-alpha-vertices}).

Therefore, this case yields a contradiction. 

\item[{\sl Case 2:}] $\quad$ There is a vertex ${\alpha_0}$ in $\sZ$ which satisfies 
\begin{equation*}
\exists {j_0} \in \sU_2 , \; ({\alpha_0}, {j_0}) \in \sE, \; 
{j_0} \mbox{ is not critical}
\mbox{ and } \sfn'(({\alpha_0}, {j_0})) = 1 \; . 
\end{equation*} 
However, by the definition of $\sZ$, there is an integer $\ell$ and a set of edges 
\begin{equation*}
 \{ ({\alpha_h}, {j_{h+1}}) \}_{h = 0, 1, \cdots, \ell} \subseteq \sE \mbox { and }
 \{ ({\alpha_h}, {j_{h}}) \}_{h = 1, 2, \cdots, \ell} \subseteq \sE \; ,
\end{equation*} 
such that ${j_{\ell+1}} = {i_0}$, 
\begin{eqnarray*}
{\alpha_h} \in \sZ & \mbox{for} & h = 0, 1, \cdots, \ell \; , \\
{j_h} \in \sU_2 & \mbox{for} & h = 1, 2, \cdots, \ell+1 \; , 
\end{eqnarray*}
and 
\begin{eqnarray*} 
 \sfn'(({\alpha_h}, {j_{h+1}})) = 0  \mbox{ for } h = 0, 1, \cdots, \ell \; , \\
 \sfn'(({\alpha_h}, {j_{h}}) = 1 \mbox{ for } h = 1, 2, \cdots, \ell \; .  
\end{eqnarray*} 

We define the flow $\sfn''$ as
\begin{equation*}
\hspace{-5ex} \sfn''(e) = 
 \left\{ \begin{array}{cl}
1 & \hspace{-1ex} \mbox{if } e \in \{ ({\alpha_h}, {j_{h+1}}) \}_{h = 0, 1, \cdots, \ell} 
\; \cup \; \{ ( {j_{\ell+1}},t) \} \\
0 & \hspace{-1ex} \mbox{if } e \in \{ ({\alpha_h}, {j_{h}}) \}_{h = 0, 1, \cdots, \ell} 
\; \cup \; \{ ( {j_{0}},t) \} \\
\sfn'(e) & \hspace{-1ex} \mbox{for all other edges $e$}
\end{array} \right. \hspace{-2ex} . 
\end{equation*}
This $\sfn''$ is a legal flow in $(\sG(\sV, \sE), \scp)$. Moreover, it has the same value as $\sfn'$, and the 
number of critical vertices $i \in \sU_2$ having 
$\sfn''((i, t)) = 1$ is strictly larger than for $\sfn'$. 
\end{description}

We conclude that there exists an integer flow $\sfn^*$ in $(\sG(\sV, \sE), \scp)$ of 
value $\sum_{\alpha \in \rrrm} k_\alpha$,
such that for every critical vertex $i \in \sU_2$, $\sfn^*((i,t)) = 1$. 
We define 
\[
a_i = \left\{ \begin{array}{cl}
\alpha & \mbox{ if } \sfn^*((\alpha, i)) = 1 \mbox{ for some } \alpha \in \sU_1 \\ 
0 & \mbox{ otherwise } 
\end{array} \right. \; . 
\]
and $\blda = (a_i)_{i \in \Gamma}$. For this selection of $\blda$, we have $\blda \in \code_\Gamma^{(\bldkk)}$
and the properties (i) and (ii) are satisfied. 
\end{proof}

\section*{Acknowledgements} 

The authors would like to thank the anonymous reviewers, as well as the associate editor I.~Sason, for their comments which improved the presentation of the paper. They would also like to thank I.~Duursma, J.~Feldman, R.~Koetter and O.~Milenkovic for helpful discussions.




\begin{thebibliography} {99}

\bibitem{Gallager}
    {R. G. Gallager,}
    {``Low-density parity-check codes,''}
    {\em IRE Transactions on Information Theory,} vol. IT-8, pp.~21--28, Jan. 1962.    

\bibitem{Wiberg}
		{N. Wiberg,}
    {\em Codes and Decoding on General Graphs.} 
    Ph.D. Thesis, Link\"oping University, Sweden, 1996. 

\bibitem{FKKR}
  {G. D. Forney, R. Koetter, F. R. Kschischang, and A. Reznik,}
  {``On the effective weights of pseudocodewords for codes defined on graphs with cycles,''}
  vol. 123 of {\em Codes, Systems, and Graphical Models,} IMA Vol. Math. Appl., ch. 5, pp.~101-112,
  Springer, 2001. 

\bibitem{KV-characterization}
  {R. Koetter, W.-C. W. Li, P. O. Vontobel, and J. L. Walker,}
  {``Characterizations of pseudo-codewords of LDPC codes,''}
  Arxiv report arXiv:cs.IT/0508049, Aug. 2005.

\bibitem{KV-IEEE-IT}
		{P. Vontobel and R. Koetter,}
    {``Graph-cover decoding and finite-length analysis of message-passing iterative decoding of LDPC codes,''}
    to appear in {\em IEEE Transactions on Information Theory,}
    Arxiv report arXiv:cs.IT/0512078, Dec. 2005.   

\bibitem{Feldman-thesis}
    {J. Feldman,}
    {\em Decoding Error-Correcting Codes via Linear Programming.}
    Ph.D. Thesis, Massachusetts Institute of Technology, Sep. 2003.    

\bibitem{Feldman}
    {J. Feldman, M. J. Wainwright, and D. R. Karger,}
    {``Using linear programming to decode binary linear codes,''}
    {\em IEEE Transactions on Information Theory,} vol. 51, no. 3, pp.~954--972, March 2005. 

\bibitem{Caire}
    {G. Caire, G. Taricco, and E. Biglieri,}
    {``Bit-interleaved coded modulation,''}
    {\em IEEE Transactions on Information Theory,} vol. 44, no. 3, pp.~927--946, May 1998.    

\bibitem{Ritcey}
    {X. Li and J. A. Ritcey,}
    {``Bit-interleaved coded modulation with iterative decoding,''}
    {\em Proc. IEEE International Conference on Communications (ICC),} vol. 2, pp.~858--863, Sep. 1999. 

\bibitem{Deepak}
    {D. Sridhara and T. E. Fuja,}
    {``LDPC codes over rings for PSK modulation,''}
    {\em IEEE Transactions on Information Theory,} vol. 51, no. 9, pp.~3209--3220, Sep. 2005.       

\bibitem{MacKay}
    {M. C. Davey and D. J. C. MacKay,}
    {``Low density parity check codes over $\GF(q)$,''}
    {\em IEEE Communications Letters,} vol. 2, no. 6, pp.~165--167, June 1998.    

\bibitem{LSLG}
    {X. Li, M. R. Soleymani, J. Lodge, and P. S. Guinand,}
    {``Good LDPC codes over $\GF(q)$ for bandwidth efficient transmission,''}
    {\em Proc. 4th IEEE Workshop on Signal Processing Advances in Wireless Communications (SPAWC),} 
    pp.~95--99, June 2003.    

\bibitem{Bennatan1}
    {A. Bennatan and D. Burshtein,}
    {``On the application of LDPC codes to arbitrary discrete-memoryless channels,''}
    {\em IEEE Transactions on Information Theory,} vol. 50, no. 3, pp.~417--438, March 2004.    

\bibitem{Bennatan2}
    {A. Bennatan and D. Burshtein,}
    {``Design and analysis of nonbinary LDPC codes for arbitrary discrete-memoryless channels,''}
    {\em IEEE Transactions on Information Theory,} vol. 52, no. 2, pp.~549--583, Feb. 2006.    

\bibitem{Bennatan-thesis}
    {A. Bennatan,}
    {\em The Application of LDPC Codes to New Problems in Communications.}
    Ph.D. Thesis, Tel Aviv University, Jan. 2007.    

\bibitem{Kelley-Sridhara-ISIT-2006}
  {C. A. Kelley, D. Sridhara, and J. Rosenthal,}
  {``Pseudocodeword weights for non-binary LDPC codes,''}
  {\em Proc. IEEE International Symposium on Information Theory (ISIT)}, Seattle, USA, pp.~1379-1383, July 2006.   

\bibitem{Forney}
  {G. D. Forney, Jr.,}
  {``Geometrically uniform codes,''}
  {\em IEEE Transactions on Information Theory}, vol.\ 37, issue 5, pp.~1241--1260, Sep. 1991.

\bibitem{Richardson}
  {T. Richardson and  R. E. Urbanke,}
  {``On the capacity of LDPC codes under message-passing decoding,''}
  {\em IEEE Transactions on Information Theory}, vol.\ 51, no.\ 9, pp.~3209--3220, Sep. 2005.

\bibitem{Flanagan_ISIT_08}
  {M. F. Flanagan,}
  {``Codeword-independent performance of nonbinary linear codes under linear-programming and sum-product decoding,''}
  {\em Proc. IEEE International Symposium on Information Theory (ISIT)}, Toronto, Canada, pp.~1503--1507, July 2008.

\bibitem{Hof_Sason_Shamai}
  {E. Hof, I. Sason, and S. Shamai (Shitz),}
  {``Performance bounds for nonbinary linear block codes over memoryless symmetric channels,''}
  {\em IEEE Transactions on Information Theory}, vol.\ 55, no.\ 3, pp.~977--996, March 2009.

\bibitem{Chertkov}
  {M. Chertkov and  M. Stepanov,}
  {``Pseudo-codeword landscape,''}
  {\em Proc. IEEE International Symposium on Information Theory (ISIT)}, Nice, France, pp.~1546--1550, June 2007.

\bibitem{Feldman-Yang}
  {K. Yang, X. Wang, and J. Feldman,}
  {``Cascaded formulation of the fundamental polytope of general linear block codes,''}
  {\em Proc. IEEE International Symposium on Information Theory (ISIT)}, Nice, France, pp.~1361--1365, June 2007.

\bibitem{Feldman-Yang2}
  {K. Yang, X. Wang, and J. Feldman,}
  {``A new linear programming approach to decoding linear block codes,''}
  {\em IEEE Transactions on Information Theory}, vol.\ 54, no.\ 3, pp.~1061--1072, March 2008.

\bibitem{Boyd}
    {S. Boyd, L. Vandenberghe,}
    {\em Convex Optimization,} 
    Cambridge: Cambridge University Press, 2004.  
 
\bibitem{Schrijver}
    {A. Schrijver,}
    {\em Theory of Linear and Integer Programming,} 
    New York: John Wiley \& Sons, 1998.  
 
\bibitem{Barvinok}
    {A. Barvinok,} 
    {\em A Course in Convexity,}
    vol. 54 of {\em Graduate Studies in Mathematics.} American Mathematical Society, Providence, RI, 2002. 

\bibitem{MacWilliams_Sloane}
    {F. J. MacWilliams and N. J. A. Sloane,} 
    {\em The Theory of Error Correcting Codes.}
    Amsterdam: North-Holland, 1977.  

\bibitem{SF}
    {V. Skachek, M. F. Flanagan,}
    {``Lower bounds on the minimum pseudodistance for linear
    codes with $q$-ary PSK modulation over AWGN,''}
    {\em Proc. 5-th International Symposium on Turbo Codes and Related Topics,} 
    Lausanne, Switzerland, September 2008.

\bibitem{Cormen}
    {T. H. Cormen, C. E. Leiserson, R. L. Rivest, and C. Stein,} 
    {\em Introduction to Algorithms.}
    Second edition, MIT Press and McGraw-Hill, 2001.  





    






\end{thebibliography}
\end{document}